\RequirePackage{amsthm}
\documentclass[smallextended]{svjour3}

\RequirePackage{fix-cm}
\usepackage{mathrsfs,amssymb,amsfonts,amsmath,url,bbm,enumerate}
\usepackage{graphicx, subfigure, color}
\usepackage[draft]{fixme}
\usepackage{enumitem}
\usepackage{hyperref}

\newif\ifCommentsOn
\CommentsOnfalse   

\ifCommentsOn  
\fxsetup{ status=draft,   author=,    layout=inline,    theme=color}
\definecolor{fxwarning}{rgb}{0.8000,0.0000,0.0000} 
\renewcommand{\comment}[1]{ \textcolor{red}{\bf{[[}}\fxwarning{#1}\textcolor{red}{\bf{]]}} }
\else 
\renewcommand{\comment}[1]{}
\fi

\newcommand{\hide}[1]{}
\renewcommand{\hat}[1]{\widehat{#1}}

\renewcommand{\tilde}[1]{\widetilde{#1}}

\newcommand{\R}{\mathbb{R}}
\newcommand{\N}{\mathbb{N}}
\newcommand{\Z}{\mathbb{Z}}

\newcommand{\Ltwo}[1]{\big\lVert #1 \big\rVert}
\newcommand{\ltwo}[1]{\lVert #1 \rVert}

\newcommand{\divconst}{{C}}  

\newcommand{\conv}{\mathrm{Conv}}
\newcommand{\maxpert}{\theta}

\newcommand{\ptraj}{\tilde{x}}

\newcommand{\ball}{\mathcal{B}}

\newcommand{\pert}{U}    
\newcommand{\ff}{\zeta}  
\newcommand{\phicor}{{\hat{\boldsymbol{\phi}}}}
\newcommand{\rcor}{\hat{\mathbf{r}}}
\newcommand{\zcor}{{\hat{\mathbf{z}}}}
\newcommand{\xcor}{{\hat{\mathbf{x}}}}
\newcommand{\ycor}{{\hat{\mathbf{y}}}}
\newcommand{\measurable}{{\cal{L}}_{n}^1}

\newcommand{\Prt}{U}
\newcommand{\tf}{ { \lfloor t\rfloor } }

\newcommand{\sigmax}{{\sigma_{\max{}}}}
\newcommand{\sigmin}{{\sigma_{\min{}}}}
\newcommand{\mumax}{{\mu_{\max{}}}}
\newcommand{\dtraj}{z}
\newcommand{\dpert}{V}
\newcommand{\ddrift}{\mu}
\newcommand{\drift}{\xi}

\newcommand{\grad}{\nabla}
\newcommand{\yt}{{\tilde{y}}}
\newcommand{\pp}{{U'}}
\newcommand{\lip}{\gamma}
\newcommand{\ord}{\mathrm{Ord}}
\newcommand{\alphamax}{{\alpha^*}}
\newcommand{\xt}{\ptraj}
\newcommand{\alphah}{\hat{\alpha}}
\newcommand{\real}{\mathrm{Re}}

\usepackage[normalem]{ulem}

\long\def\oli#1{{\color[rgb]{0,.8,.8}[#1]}}
\newcommand\redsout{\bgroup\markoverwith{\textcolor{red}{\rule[0.5ex]{2pt}{.5pt}}}\ULon}

\renewcommand{\oli}[1]{}

\theoremstyle{plain}
\newtheorem{clm}{Claim}

\theoremstyle{definition}

\makeatletter
\newtheoremstyle{examples}
  {3pt}
  {3pt}
  {\addtolength{\@totalleftmargin}{3em}
   \addtolength{\linewidth}{-5em}
   \parshape 1 2.5em \linewidth}
  {}
  {\bfseries}
  {.}
  {.5em}
  {}
\makeatother

\theoremstyle{examples}

\begin{document}

\title{When do Trajectories have Bounded Sensitivity to Cumulative Perturbations?}
\author{ 
	Arsalan Sharifnassab 
	and
	S. Jamaloddin Golestani
	\vspace*{-30pt}
}
\institute{A. Sharifnassab \at
	Department of Electrical Engineering, Sharif University of Technology, Tehran, Iran \\
	\email{a.sharifnassab@gmail.com}           
	\and
	S. J. Golestani \at
	Department of Electrical Engineering, Sharif University of Technology, Tehran, Iran \\
	\email{golestani@sharif.edu}  
}

\date{Received: date / Accepted: date}

\maketitle

\begin{abstract}
We investigate sensitivity to cumulative perturbations for a few  dynamical system classes of practical interest.
A system is said to have bounded sensitivity to cumulative perturbations (bounded sensitivity, for short) if  an additive disturbance leads to a change in the state trajectory that is bounded by a constant multiple of the size of the cumulative disturbance. 
As our main result, we show that there exist dynamical systems in the form of (negative) gradient field  of a convex function that have unbounded sensitivity. 
We show that the result holds even when the convex potential function is piecewise linear.
This resolves a question raised in \cite{AlTG19sensitivity}, wherein it was shown that
the (negative) (sub)gradient field  of a piecewise linear and convex function  has  bounded sensitivity if the number of linear pieces is finite. 
Our results establish that the finiteness assumption is indeed necessary.

Among our other results, we 
provide a necessary and sufficient condition for a linear dynamical system to have bounded sensitivity to cumulative perturbations. 
We also establish that the bounded sensitivity property is preserved, when a dynamical system with bounded sensitivity undergoes certain transformations.
These transformations include convolution, time discretization, and spreading of a system (a transformation that captures approximate solutions of a system).
\end{abstract}

\hide{
}

\section{\bf Introduction}
We study a property of dynamical systems, that when satisfied, provides a bound on sensitivity of the state trajectory  with respect to additive disturbances. 
Consider a dynamical system of the form
$$\dot{x}(t)=f\big(x(t)\big), $$
and its perturbed counterpart 
\begin{equation} \label{eq:dt3}
\frac{d}{dt}\xt(t)=f\big(\xt(t)\big)+u(t).
\end{equation}
Here, $x(t)$ and $u(t)$ take values in $\R^n$.
In order to motivate our results, lets temporarily assume that the system is nonexpansive, in the sense that for any  solution $y(\cdot)$ of $\dot{y}(t)=f\big( y(t)\big)$, and any pair of times $t_1$ and $t_2$ with $t_2\ge t_1$,
$$\| y(t_2)-x(t_2)\| \,\leq\, \| y(t_1)-x(t_1)\|,$$
for a given norm $\|\cdot\|$. 
In this case, assuming the same initial conditions, $\xt (0)= x(0)$, a simple integration yields a bound of the form
\begin{equation} \label{eq:dt1}
\big\| \xt (t)- x(t)\big\| \,\leq\, 
\int_{0}^{t} \big\| u(s)\big\|\,ds.
\end{equation}
However, our goal is to derive stronger bounds, of the form
\begin{equation} \label{eq:dt2}
\big\| \xt (t)- x(t)\big\| \,\leq\, 
C\, \sup_{\tau<t} \Big\|\int_{0}^{\tau} u(s)\, ds \Big\|,
\end{equation} 
for some constant $C>0$ independent of $u(\cdot)$. Property \eqref{eq:dt2} is referred to as bounded sensitivity to cumulative perturbations.

A bound of the form \eqref{eq:dt2} is not valid in general. 
\hide{
In Section \ref{sec:example}, we give examples of nonexpansive systems, and examples of gradient fields of convex functions, for which no constant $C$  satisfies \eqref{eq:dt2}.
On the other hand, there are (linear) systems that are not nonexpansive, yet admit a bound of the form \eqref{eq:dt2}, cf.~ Section~\ref{sec:linear}.
}
However, it is shown in \cite{AlTG19sensitivity} that a bound of type \eqref{eq:dt2} is valid for the class of Finitely Piecewise Constant Subgradient (FPCS) systems.
An FPCS system is, by definition, the (negative) gradient field of a piecewise linear and convex function with finitely many pieces.
It is shown in \cite{AlTG19stokes} that FPCS systems actually contain the seemingly larger class of nonexpansive \emph{finite-partition} systems.
Finite-partition systems are dynamical systems that have a constant drift over  each of the finitely many regions that form a partition of $\R^n$.
Such systems are common in control, when dealing with hybrid systems with a finite
set of control actions, that can be applied in certain parts of the state space.
Examples include communication networks \cite{TassE92,Neel10}, processing systems \cite{RossBM15}, manufacturing systems and inventory management \cite{PerkS98,Meyn08}, etc.

A bound of type \eqref{eq:dt2} is particularly useful in dealing with systems driven by stochastic noise.
Under  usual 
probabilistic assumptions, $\sup_{\tau<t}\big\|\int_{0}^{\tau} u(s)\,ds\big\|$ 
roughly grows  as $\sqrt{t}$, whereas  
$\int_{0}^{t} \|u(\tau)\| \,d\tau$ 
grows at the rate of $t$, 
with high probability.
See  \cite{AlTG19ssc} for applications of  this bound to the analysis of the celebrated Max-Weight policy for real-time job scheduling \cite{TassE92}.

\hide{
As such, for a perturbed system with stochastic noise $u(\cdot)$, a bound of type \eqref{eq:dt2}, if valid, provides a decent bound on the effect of noise on the state trajectory.
A notable example is the dynamics of the celebrated Max-Weight algorithm for real-time job scheduling \cite{TassE92}.
Relying on the bounded sensitivity of FPCS systems \cite{AlTG19sensitivity}, we establish in \cite{AlTG19ssc}  a similar bound for the dynamics of the Max-Weight policy.
There, we show that the queue lengths of the stochastic system under Max-Weight policy stay decently close to the \emph{fluid model solutions}.
Such strong sensitivity bounds enabled us to partially establish in \cite{AlTG19ssc} an open conjecture from \cite{ShahW12} on \emph{state space collapse}, 
as well as an open conjecture from \cite{MarkMT17} on  stability in presence of heavy-tailed traffics, as will  be reported elsewhere.
}

\hide{
A popular approach for the analysis of control systems is via the use of fluid models\footnote{For 
a discrete time system $x^+=x+f(x)$ defined over the entire $\R^n$,  the continuous time system $\dot{x}=f(x)$ is referred to  as a corresponding fluid model.}.
In order to avoid complications of stochastic or perturbed discrete time systems, one can approximate the discrete time system with its fluid model, 
analyse the property of interest for the fluid model, and use those results to infer properties such as stability \cite{Dai95,DaiM95,Gama00}
for the perturbed discrete time system.
Fluid models have been a major tool not only for the analysis of system properties but also for  synthesis of control algorithms. 
Again, to avoid the computational intractability of solving dynamic programming \cite{Stid85,MaCM10} for optimal control decisions in stochastic systems,
one often approximates the stochastic system with its fluid model \cite{ChenDM04,BertNP15},
solves for the optimal solutions of the fluid model \cite{Meyn97,Meyn05},
and translates these solutions via discrete review methods to use them for controlling the stochastic system \cite{Magl99,Magl00,FleiS05,Meyn08}.
In this framework, the quality of fluid model approximation,  captured in the sensitivity bounds, is central to the performance of the synthesized control methods.
}

In this paper, we investigate the extent to which a bound of type \eqref{eq:dt2} can or cannot generalize  to other classes of dynamical systems of practical interest.
We consider linear systems and  derive a necessary and sufficient condition for them to have bounded sensitivity. 
In particular, we show that a linear system admits a bound of the form \eqref{eq:dt2} if and only if it is stable and has no closed orbit.
More importantly, we show that the gradient field of a strictly convex function can have unbounded sensitivity.
In the same spirit, we provide an example of the gradient field of a piecewise linear convex function with infinitely many pieces, for which \eqref{eq:dt2} does not hold.
The two latter results are quite counter-intuitive: while the (negative) subgradient field of a piecewise linear convex function with finitely many pieces has bounded sensitivity,  the (negative) subgradient field of some continuously differentiable (or even infinitely piecewise linear) convex functions can have unbounded sensitivity.
These examples shed some light on the limitations of extending the sensitivity bound to generalizations of FPCS systems, and also on the inevitable complications of any proof for bounded sensitivity of these (FPCS) systems; cf. Section \ref{sec:discussion} for a detailed discussion.


We also study some transformations that when applied on a dynamical system with bounded sensitivity, preserve the bounded sensitivity property.
In particular, we show for any continuous time dynamical system with bounded sensitivity that its analogous discrete time system  also has bounded sensitivity.
We establish a similar result when the dynamical system is convolved by a kernel, and when the system is \emph{spread}, 
that is allowing for the trajectories to move along the drifts of nearby points. 
These results provide grounds on which the proofs of our main result (on unbounded sensitivity of the gradient field of a strictly convex function) relies.

A seemingly relevant literature is the \emph{input-to-state stability} \cite{JianW01,MarrAC02,Ange04,Sont08,Sont96,CaiT13,AngeSW00}.
As discussed in Section 1 of \cite{AlTG19sensitivity}, for a system with additive disturbance, $\dot{x}(t)=f\big(x(t)\big)+u(t)$, 
the input-to-state stability and a bound of the form \eqref{eq:dt2}
do not imply one another. 
We refer the reader to \cite{AlTG19sensitivity} for a comprehensive discussion on the relation of the sensitivity notion of type \eqref{eq:dt2} to the seemingly relevant literatures.

The rest of the paper is organized as follows.
We begin with formal definitions and preliminaries in Section~\ref{sec:pre}.
We then present our main results in Sections~\ref{sec:linear}, \ref{sec:main bis of systems}, and~\ref{sec:main operations}.
In Section~\ref{sec:linear} we give necessary and sufficient conditions for bounded sensitivity of linear systems.
In Section~\ref{sec:main bis of systems} we investigates sensitivity of gradient fields of convex functions, and provide examples of differentiable (as well as piecewise linear) convex functions whose subgradient fields have unbounded sensitivity.
In Section~\ref{sec:main operations}, we study transformations on dynamical systems that preserve boundedness of sensitivity,
and set the stage and provide the required machinery for the proofs of the results of Section~\ref{sec:main bis of systems}. 	
We give the proofs of our main results in Sections \ref{sec:proof th linear}, \ref{sec:proof th spread sys}, \ref{app:proof unbounded sensitivity of F}, \ref{app:pwc example}, and~\ref{sec:proof disc} while relegating some of the details to the appendix, for improved readability. 
Finally, we discuss our results as well as several open problems and  directions of future research in Section~\ref{sec:discussion}.



\medskip
\section{\bf Preliminaries} \label{sec:pre}

As in \cite{Stew11}, we  identify
a dynamical system with a set-valued function $F:\R^n\to 2^{\R^n}$
and the associated differential inclusion $\dot{x}(t)\in F(x(t))$. We start with a formal definition, which allows for the presence of perturbations.

\begin{definition}[Perturbed Trajectories]\label{def:integral pert traj} 
Consider a dynamical system $F:\R^n\to 2^{\R^n}$, and let $\Prt:\R\to\R^n$ be a right-continuous function, which we refer to as the \emph{perturbation}.
 Suppose that there exist measurable and integrable functions $\ptraj({\cdot})$ and $\ff({\cdot})$ of time that satisfy \begin{equation}\label{eq:def of integral pert}
\begin{split}
\ptraj(t) &= \int_0^t \ff(\tau)\,d\tau \,+\, \Prt(t),\qquad \forall\ t\ge 0,\\
\ff(t)&\in F\big(\ptraj(t)\big),\qquad \forall\ t\ge 0.
\end{split}
\end{equation}
We then call $\Prt$ the \emph{perturbation},
and such $\ptraj$ and $\ff$ are called a
\emph{perturbed trajectory}
and a \emph{perturbed drift},  respectively. In the special case where 
$U$ is  identically zero, we also refer to $\ptraj$  as an \emph{unperturbed trajectory}.

\end{definition}

We now define two notions of bounded sensitivity, the second of which implies the first.

\begin{definition}[Bounded  Sensitivity] \label{def:bpis}
A dynamical system $F(\cdot)$ is said to have  bounded sensitivity if there exists a constant $\divconst$ such that for any unperturbed trajectory $x(\cdot)$, any perturbation function $\pert(\cdot)$, 
and its corresponding  perturbed trajectory  $\ptraj(\cdot)$ initialized at $\ptraj(0)=x(0)$,
\begin{equation} \label{eq:bis}
\Ltwo{\ptraj(t)-x(t)} \,\le \, \divconst\, \sup_{\tau\le t} \Ltwo{\pert(\tau) }, \qquad \forall t\ge0.
\end{equation}
Further, if for any pair  $\pert_1(\cdot)$ and $\pert_1(\cdot)$ of perturbation functions and their corresponding  perturbed trajectories  $\ptraj_1(\cdot)$ and $\ptraj_2(\cdot)$, initialized at $\ptraj_1(0)=\ptraj_2(0)$,
\begin{equation} \label{eq:bpis}
\Ltwo{\ptraj_1(t)-\ptraj_2(t)} \,\le \, \divconst\, \sup_{\tau\le t} \Ltwo{\pert_1(\tau) -\pert_2(\tau)}, \qquad \forall t\ge0,
\end{equation}
then $F(\cdot)$ is said to have bounded sensitivity in strong sense. 
\end{definition}

A bound of type \eqref{eq:bpis} implies  the bound in \eqref{eq:bis}, by simply letting one of the perturbation functions equal to zero.

Throughout the paper we often 
assume existence of a constant $\lip$, for which 
\begin{equation} \label{eq:constant not to blow up}
\Ltwo{F(x)}\le \gamma\big(1+\lVert x \rVert_2\big),\qquad \forall x\in \R^n
\end{equation}
This assumption  is to prevent the solutions from blowing up in finite time.

\hide{
A dynamical system $F(\cdot)$ is called nonexpansive if for any pair  $x(t)$ and $y(t)$ of  trajectories, and for any time $t_1$ and any $t_2\ge t_1$,
\begin{equation}
\Ltwo{x(t_1)-y(t_1)}\le \Ltwo{x(t_2)-y(t_2)}.
\end{equation} 
}

We say that a dynamical system $F(\cdot)$ is a \emph{subgradient dynamical system} if there exists a convex function $\Phi(\cdot)$, such that for any $x\in\R^n$, $F(x)=-\partial\Phi(x)$, 
where $\partial\Phi(x)$ denotes the subdifferential of $\Phi$ at $x$.
If further, $\Phi(x)$ is of the form  
$$\Phi(x)=\max_{i}\big(-\mu_i^Tx+b_i\big),$$ 
for some $\mu_i\in\R^n$, $b_i\in\R$, and with $i$ ranging over a {\bf finite} set,
we say that $F$ is a \emph{Finitely Piecewise Constant Subgradient} (FPCS, for short) system.
It is shown in \cite{AlTG19sensitivity} that any FPCS system has  bounded sensitivity.

\begin{theorem}[Theorem 1 of \cite{AlTG19sensitivity}] \label{th:p1 sensitivity}
Every FPCS system has bounded sensitivity  in the sense of \eqref{eq:bis}.
\end{theorem}

In the rest of this section we briefly discuss quasi-convexity.
A function $f:\R^n\to\R$ is said to be \emph{quasi-convex} if all of its sub-level-sets are convex sets. 
Equivalently,  for any $x,y\in \R^n$ and any $\lambda\in (0,1)$,
\begin{equation} \label{eq:def of quasi convexity}
f\big(\lambda x+(1-\lambda)y    \big) \le \max\big( f(x),f(y)  \big).
\end{equation}
It further,  for any $x,y\in \R^n$ and any $\lambda\in (0,1)$, \eqref{eq:def of quasi convexity} holds with strict inequality,
then  $f$ is said to be \emph{strictly quasi-convex}.
Under some mild assumptions, one can \emph{convexify} any strictly quasi-convex function \cite{ConnR12}:
\begin{lemma}[Corollary~1 of \cite{ConnR12}] \label{lem:convexification}
For any continuously twice differentiable and strictly quasi-convex function $f:\R^n\to \R$ with compact level-sets, 
there exists an increasing and continuously twice differentiable function $h:\R\to\R$ such that $h\circ f$ is strictly convex.
\end{lemma}

\hide{
The next definition captures the existence and uniqueness of perturbed trajectories in a  dynamical system.
\begin{definition}[Solvable Dynamical System] \label{def:completeness}
We say that a dynamical system is \emph{solvable}, if for any perturbation function $\pert(t)\in\measurable$ and for any initial condition $x(0)\in\R^n$, the  differential inclusion in (\ref{eq:pert traj}) has a solution. Furthermore, if such a solution is unique, then the dynamical system is called \emph{uniquely complete}.
\end{definition}
}

Finally, we denote by  $\R_+$ and $\Z_+$   the sets of non-negative real numbers and non-negative integers, respectively.


\section{\bf Sensitivity of Linear Systems} \label{sec:linear}
In this section we present a necessary and sufficient condition for linear dynamical systems to have bounded sensitivity. 
A linear dynamical system is a system of the form $\dot{x}=Ax$, defined in terms of a square matrix $A$.
Before going over our result for linear systems, we define a property for general dynamical systems.
\begin{definition}[Stable and Orbit-Free Systems]
A dynamical system is said to be \emph{stable} if every unperturbed trajectory stays in a bounded region, and it is \emph{orbit-free} if no unperturbed trajectory is a periodic orbit.
We use the shorthand \emph{SOF} for a stable and orbit-free system.
\end{definition}
The following lemma provides well-known facts on the stability of linear systems.
\begin{lemma}\label{lem:SOF linear}
A linear  system $\dot{x}=Ax$ is SOF if and only if its eigenvalues are either zero or have negative real parts, and the multiplicity of the zero eigenvalue equals the demansion of the associated eigenspace.
\end{lemma}

The next theorem shows that a linear system has bounded sensitivity if and only if it is SOF.

\begin{theorem}[Sensitivity of Linear Systems]\label{th:linear}
A linear dynamical system has bounded sensitivity in the (strong) sense of \eqref{eq:bpis}, if and only if it is SOF. 
Further, every non-SOF linear system has unbounded sensitivity in the sense of \eqref{eq:bis}.
Moreover, for an SOF dynamical system in  which   $A$ is diagonalizable (of the form $A=P\Lambda P^{-1}$),   \eqref{eq:bpis} is satisfied by the following constant,
\begin{equation} \label{eq:constant C for linear in diagonalizable case}
\divconst \,=\, 1\,+\, \frac{\sigmax}{\sigmin} \sum_{i}\frac{|\lambda_i |}{\big\lvert  \real(\lambda_i)  \big\rvert},
\end{equation}
where $\lambda_i$'s are eigenvalues of $A$, and  $\sigmax$ and $\sigmin$ are the largest and smallest singular values of $P$, respectively. 
In the special case that $A$ is symmetric, \eqref{eq:constant C for linear in diagonalizable case} is simplified to $\divconst=n+1$.
\end{theorem}

The proof is given in Section \ref{sec:proof th linear}.
The proof relies on a closed form expression for the solutions of  linear differential equations and  the Jordan normal form of the underlying matrix $A$.

Note that stability of a linear system is not sufficient for bounded sensitivity. According to Theorem~\ref{th:linear}, a linear system with bounded sensitivity is necessarily both stable and orbit-free.
Fig.~\ref{fig:rotary} shows how the existence of a periodic orbit can cause unbounded sensitivity in a  stable linear system.
\begin{figure} 
	\begin{center}
		{
			\includegraphics[width = .35\textwidth]{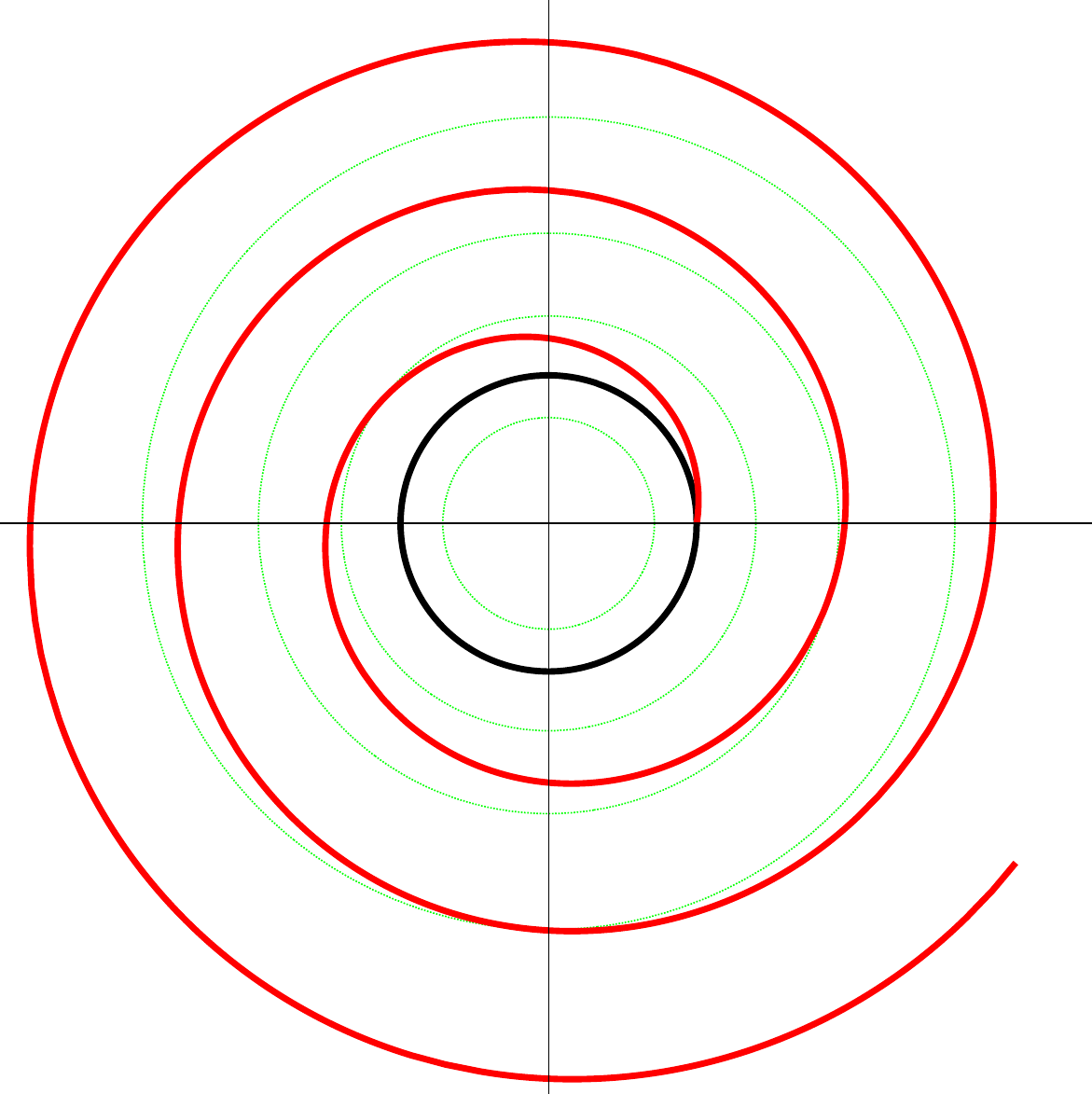}}
	\end{center}
	\caption{Trajectories of the linear dynamical system  $\dot{x}=\left(-x_2,x_1\right)$, where $\left(x_1,x_2\right)$ is the representation of $x(t)$ in the Cartesian coordinates. This dynamical system is linear and stable, and its trajectories are circular orbits centered at the origin. 
	The system has unbounded sensitivity since a bounded perturbation can cause unbounded distance between perturbed and unperturbed trajectories.  Black circle: unperturbed trajectory $x(t)=\left(\cos(t),\sin(t)\right)$. Red spiral: perturbed trajectory $\ptraj(t)=\left((t+1)\cos(t),(t+1)\sin(t)\right)$ corresponding to the bounded cumulative perturbation $U(t)=\left(\sin(t),1-\cos(t)\right)$, all in Cartesian coordinates.} 
	\label{fig:rotary}
\end{figure}

\medskip


\section{\bf  Sensitivity of Subgradient Dynamical Systems} \label{sec:main bis of systems}

In this section, we show that there exist dynamical systems driven by gradient of a convex  function that have unbounded sensitivity. The results of this section shed light on the limitations and challenges of extending Theorem~\ref{th:p1 sensitivity} to larger classes of systems. 
More concretely, the results suggest that  Theorem~\ref{th:p1 sensitivity}  will no longer hold true if we remove/weaken any of its assumptions.
The section is organized into two subsections.
In the first subsection, we present example of a  dynamical system driven by (negative) gradient of a strictly convex function. We show that this system has unbounded sensitivity. 
We then show via  another example in Subsection~\ref{subsec:pwc with infinite pieces} that unbounded sensitivity is possible even when the convex potential function is piecewise linear  with  infinitely many pieces.

\hide{
We start off by a simple example of a  nonexpansive system with unbounded sensitivity in Example~\ref{ex:rotary}. 
We  then present examples of gradients fields of convex functions that have unbounded sensitivity in Examples~\ref{ex:subdiff dont work} and~\ref{ex:pwc with infinite pieces}. 
The proofs of the lemmas of this subsection are relegated to the appendix.
}

For the ease of presentation, throughout this section we will be working with cylindrical coordinates. 
For a trajectory $x:\R\to\R^3$, and for a time $t\in\R_+$, we represent the location of $x(t)$ in  cylindrical coordinates  by $\big(r,\phi, z\big)$, which equals $\big(r\cos \phi,\, r\sin \phi, z \big)$ in the Cartesian coordinates.

\hide{
In our first example, we consider a  nonexpansive dynamical system and show that it has unbounded sensitivity. 
The example is simple by itself and is an immediate consequence of Theorem~\ref{th:linear}.
Having said that, we include it here to build intuition and set the stage  for the more elaborate examples that follow.

\begin{example}(A nonexpansive system with unbounded sensitivity) \label{ex:rotary}
Consider the two-dimensional linear dynamical system $\dot{x}=F(x)=r\phicor$ in the polar coordinates.
The  system is nonexpansive, and its trajectories are circular orbits centered at the  origin. 
Fig.~\ref{fig:rotary} shows an illustration of the trajectories of this system.
Consider an unperturbed trajectory $x(t)=(1,t)$ in the polar coordinates.
For any $t\in\R_+$, let $\ptraj(t)=(t+1,\,t)$ in the polar coordinates, and $u(t)=1\rcor$ be the differential perturbation in the local polar coordinates at $\ptraj(t)$.
Then, for any $t\in\R_+$, $\tfrac{d}{dt} \ptraj(t) =1\rcor+(t+1)\phicor= F\big(\ptraj(t)\big) +u(t)$. 
Therefore, $\ptraj(t)$ is a perturbed trajectory corresponding to perturbation $\pert(t) = \int_{0}^t u(\tau)\,d\tau$.
Moreover, $\Ltwo{U(t)} = \Ltwo{ \int_{0}^t u(\tau)\,d\tau}\le \Ltwo{ \int_{0}^t \cos(\tau)\,d\tau}+\Ltwo{ \int_{0}^t \sin(\tau)\,d\tau}\le 2$, which is bounded.
However, $\Ltwo{\ptraj(t)-x(t)} = t$, which grows unbounded for large $t$.
Therefore, no constant $C$ can satisfy \eqref{eq:bis}, and  the system has unbounded sensitivity.
\end{example}
}

\hide{
In particular, in Section \ref{sec:3d simple}, we start with a simple case, as a prologue to the main result/examples in the subsequent sections. 
There, we present a quasi-convex function, and show that its subgradient field\footnote{Note that subgradient is not defined for quasi-convex function, and we are abusing the terminology.} has unbounded sensitivity. Later, in Section \ref{sec:3d examples}, we elaborate on the same idea to construct convex functions whose subgradient fields have unbounded sensitivity. These constructions require some machineries, which we will develop in Section \ref{sec:machinery}. The results of Section \ref{sec:machinery} are also of independent interest.

\medskip

\subsubsection{\bf nonexpansive and Contractive Systems} 
In this section, we present examples of simple nonexpansive and contractive systems with unbounded sensitivity.

It follows from the definition of sensitivity that the distance of any two unperturbed trajectories can never increase. Therefore, every system with bounded sensitivity is nonexpansive. 
\medskip
In our next example, we take one step  further and study dynamical systems with contraction.
We say that a dynamical system is \emph{contractive} if $x(t)\ne y(t)$ implies $\frac{d}{dt}\Ltwo{x(t)-y(t)}<0$, for all unperturbed trajectories $x(\cdot)$ and $y(\cdot)$.
It turns that contraction is not a sufficient condition for boundedness of sensitivity:

\begin{example} \label{ex:contractive}
By an slight modification of the rotary system, consider the dynamical system $\dot{x}=-r^2 \rcor + r\phicor$. Fix some $\epsilon>0$.
The trajectory $x(t)=\big(t+1/\epsilon\big)^{-1}\rcor+t\phicor$ is an unperturbed trajectory, and  $\ptraj(t)=\epsilon\rcor+t\phicor$ is a perturbed trajectory  corresponding to the differential perturbation $\pert(t)=\epsilon^2\rcor+t\phicor$.  These trajectories are illustrated in Fig. \ref{fig:contractive}. Therefore, $\Ltwo{\ptraj(t)-x(t)} = {\epsilon}/\big(1+1/(\epsilon t)\big)\to\epsilon$ as $t$ goes to infinity, while  perturbation is upper bounded: $\Ltwo{U(t)}=\Ltwo{ \int_{0}^t \pert(\tau)\,d\tau}\le 2\epsilon^2$. Then,  $\Ltwo{\ptraj(t)-x(t)}/\sup_t \Ltwo{U(t)\,d\tau} \simeq 2/\epsilon$, which is unbounded as $\epsilon$ tends to zero. As a result, the system has unbounded sensitivity.
\end{example}

\medskip

We can yet ca dynamical system to be  \emph{strongly contractive} if there exists a constant  $\alpha$, such that for any two unperturbed trajectories $x(t)$ and $y(t)$, we have $\frac{d}{dt}\Ltwo{x(t)-y(t)}\le - \alpha \Ltwo{x(t)-y(t)}$, whenever $x(t)\ne y(t)$. 

In our previous reports, we made the following conjecture:
\begin{conjecture} \label{conj:strong contractive}
Every strongly contractive dynamical system has bounded sensitivity.
\end{conjecture}
We will disprove this conjecture in Section \ref{sec:3d examples} (cf. Remark \ref{rem:strongly contractive not sufficient}).


\subsubsection{\bf Warming UP for Subgradient Systems} \label{sec:3d simple}

}

\medskip


\subsection{\bf Gradient field of a strictly convex function with unbounded sensitivity}\label{subsec:subdiff dont work}
Consider the half cylinder  
\begin{equation}\label{eq:def Omega}
\Omega = \big\{(r,\phi,z) \,\big|\, r\le 1/4, \, z \le -1 \big\},
\end{equation}
in cylindrical coordinates. 
Let $f(r,\phi,z)$ be a solution of the following equation over $\Omega$:
\begin{equation} \label{eq:book f (smooth)}
f\,+\, z \,-\, \frac1f\, \ln\Big(\cosh \big(fr\sin(f-\phi)\big)\Big)\, -\,  \frac{r^2}{1+f}=0.
\end{equation}
The following lemma shows that $f$ is well-defined and is strictly quasi-convex.
\begin{lemma}\label{lem:properties of f}
\begin{itemize}
\item[a)]
For any $(r,\phi,z)\in\Omega$, there is a unique $f\ge 1$ that satisfies \eqref{eq:book f (smooth)}.
\item[b)]
For any $\alpha\ge 1$, the level-set $f(r,\phi,z)=a$ is a surface of the form 
\begin{equation}\label{eq:surf eq}
z(r,\phi) = -a + \frac1a \ln\Big(\cosh \big(ar\sin(a-\phi)\big)\Big)  +\frac{r^2}{1+a}.
\end{equation}
\item[c)]
$f$ is a smooth and strictly quasi-convex function, and its  level-sets are compact. 
\end{itemize}
\end{lemma}
The proof of the lemma is given in Appendix~\ref{app:proof wellformedness of f}.

For the intuition behind the definition of $f$, note that  for sufficiently large values of $f$ (when $z$ goes to $-\infty$), we have  $\ln\Big(\cosh \big(fr\,\sin(f-\phi)\big)\Big)\,/f \approx r\big|\sin(f-\phi)\big|$. 
Then,  \eqref{eq:surf eq} implies that for sufficiently large values of $a$, the level set $f=a$ is very close to the surface $z(r,\phi) = -a + r\big\lvert\sin(a-\phi)\big\rvert$. This surface has the shape of an opened book, and the books rotate as $a$ varies. 
An illustration of different level-sets of $f$ is shown in Fig.~\ref{fig:3d example}~(a). 
In light of the rotating books analogy, we can show that the gradient field of $f$ admits spring-shaped unperturbed trajectories and diverging spiral-shaped perturbed trajectories of the forms depicted in Fig.~\ref{fig:3d example}~(b). 
Having discussed the insight, we proceed to a rigorous proof and construction of the desired convex function.

\begin{figure} 
\begin{center}
\subfigure[] {\includegraphics[width = .45\textwidth]{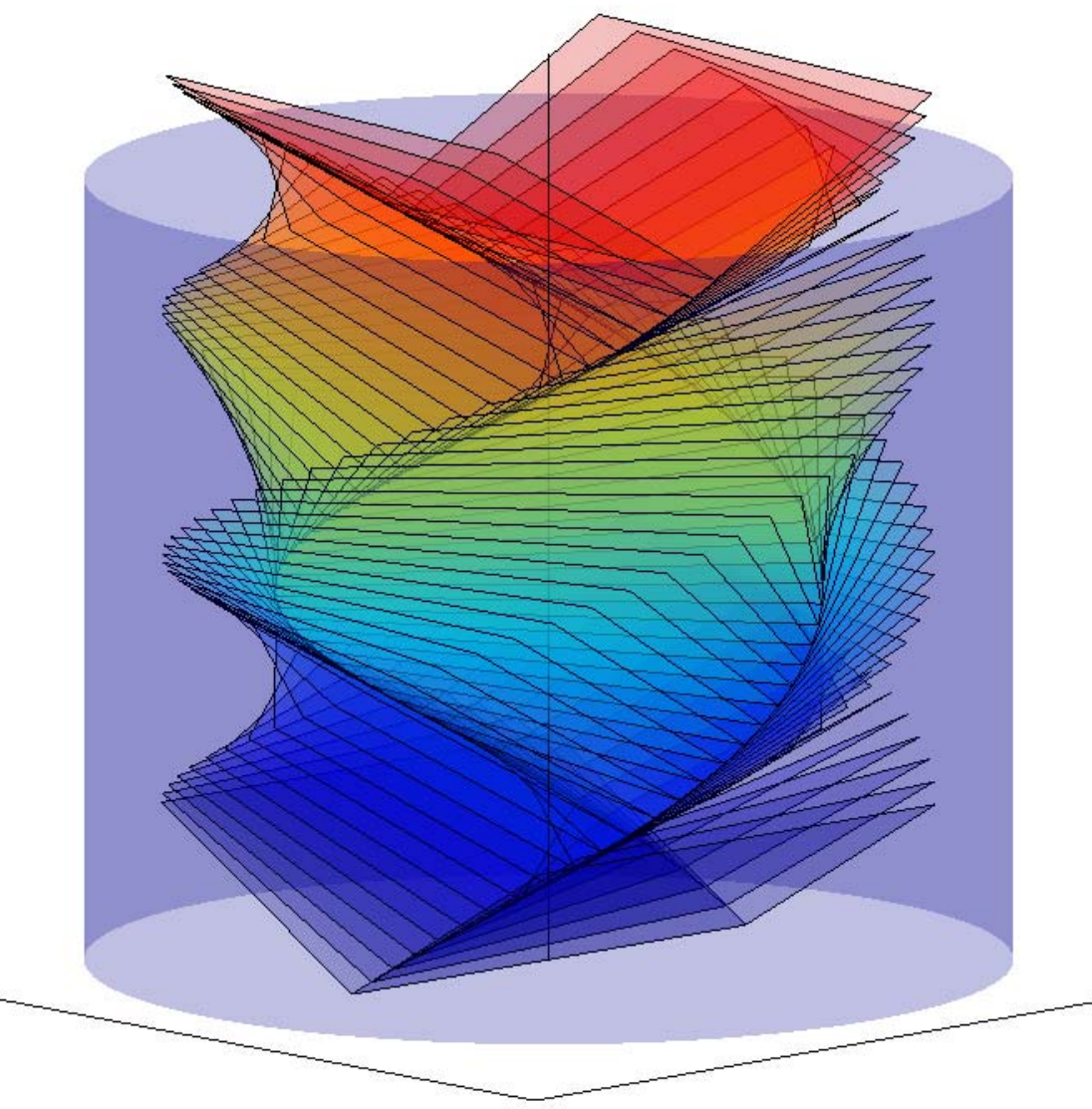}}\quad
\subfigure[] {\includegraphics[width = .45\textwidth]{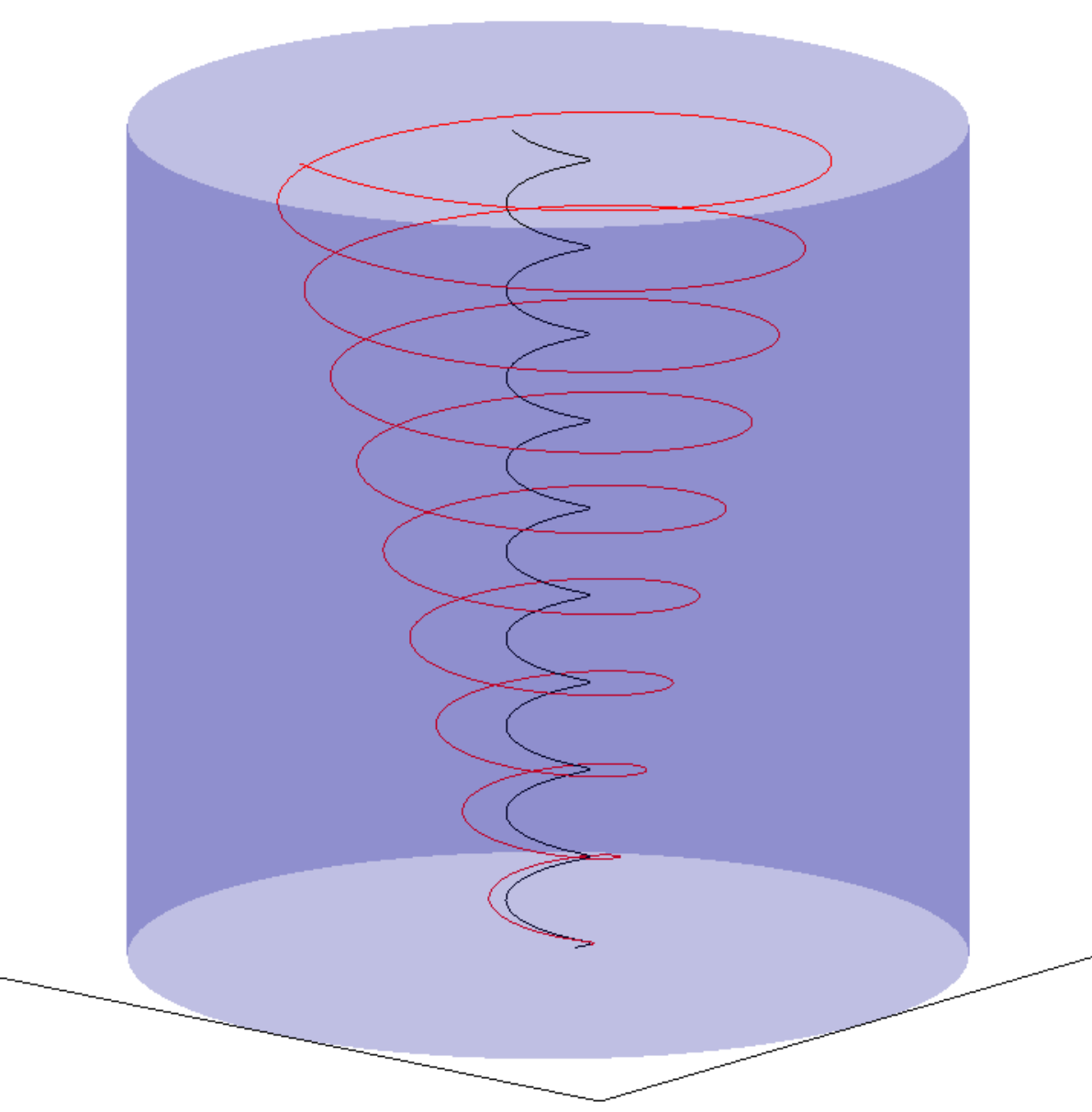}}
\vspace{-.5cm}
\end{center}
\caption{The dynamical system of Section~\ref{subsec:subdiff dont work} that has unbounded sensitivity. (a) shows the level sets of the quasi-convex  function $f$ defined in \eqref{eq:book f (smooth)} and (b) depicts a pair of perturbed (red) and unperturbed (black) trajectories (traversing the curves upwards). 
}
\label{fig:3d example}
\end{figure}

It follows from Lemma~\ref{lem:properties of f}~(c) and Lemma~\ref{lem:convexification} that there exists an increasing and twice continuously  differentiable function $h:\R\to \R$ such that $ h\circ f$ is a strictly convex function. 
Let $\Phi\triangleq h\circ f$ and $F$ be the gradient field of $\Phi$.
Then,

\begin{theorem} \label{th:3d F has unbounded sensitivity}
 $F$ has unbounded sensitivity, even when its domain is restricted to the region 
 \begin{equation}\label{eq:def Omega zeta}
\Omega_\zeta \triangleq \big\{(r,\phi,z) \,\big|\, r\le 1/6, \, z < \zeta \big\},
\end{equation}
for all $\zeta\le -1$.
\end{theorem}

\oli{We need the restricted domain, later, in the proof of Theorem \ref{th:pwc example}.}

The proof of the theorem is given in Section~\ref{app:proof unbounded sensitivity of F}, and goes by showing that $F$ admits trajectories of the form depicted in Fig.~\ref{fig:3d example}~(b). 
In the course of the proof we need a  machinery that we develop later on in Section~\ref{sec:spread}.




\subsection{\bf Subgradient field of a piecewise linear convex function that has unbounded sensitivity} \label{subsec:pwc with infinite pieces}
Here we capitalize on the example of Subsection~\ref{subsec:subdiff dont work} to construct a dynamical system driven by the (negative) gradient of a piecewise linear and convex function that has unbounded sensitivity. 

The high-level idea is as follows. Consider the convex potential function $\Phi$ of Subsection~\ref{subsec:subdiff dont work}, whose gradient field $F$ has unbounded sensitivity.
We construct a piecewise linear approximation $\Psi$ of $\Phi$ with infinite number of pieces, such that the approximation error tends to zero as $z$  goes to $-\infty$.  
To do this, we consider a fine grid within the half-cylinder $\Omega$, defined in \eqref{eq:def Omega}, with increasing resolution as $z$ goes to $-\infty$, and a corresponding triangulation of $\Omega$ with simplexes. 
We then let $\Psi(p)=\Phi(p)$ on the grid points $p$, and let $\Psi$ be the linear interpolation inside each simplex. 
The resulting $\Psi$ is a convex function. 
Let $H$ be the (sub)gradient field of $\Psi$.
Since the resolution of grid points increases as $z$ goes to $-\infty$, for sufficiently small values of $z$, $H$ would give a good approximation of $F$, and we can use Theorem~\ref{th:3d F has unbounded sensitivity} to deduce unbounded sensitivity also for $H$.
We now state the main result in the following theorem, and leave the detailed construction and the proof to Section~\ref{app:pwc example}.

\begin{theorem} \label{th:pwc example}
There exists a piecewise linear convex function whose (sub)gradient field has unbounded sensitivity.
\end{theorem}

Theorem~\ref{th:pwc example} shows that a piecewise constant subgradient field with infinitely many pieces can have unbounded sensitivity.
This is in contrast to Theorem~\ref{th:p1 sensitivity}, according to which the sensitivity of  any piecewise constant subgradient field with a finite number of pieces is bounded.
We conclude that the assumption of finiteness of the number of pieces  cannot be relaxed in Theorem~\ref{th:p1 sensitivity}. 


\medskip
\section{\bf Transformations that Preserve bounded sensitivity}\label{sec:main operations}
In this section, we study transformations on a dynamical system that preserve  bounded sensitivity.
In particular, we show for any dynamical system with bounded sensitivity that \emph{discretization of time} 
and \emph{spreading} will preserve  bounded sensitivity, up to an additive constant.
The section comprises two subsections, each devoted to one of these transformations.
The result of the first subsection on spreading a system (Theorem~\ref{th:spread sys}) shows that for a system with bounded sensitivity,  perturbations have bounded effect on its ``approximate solutions'' (see Subsection~\ref{sec:spread}), as well. 
This result not only provides insight into the sensitivity  of approximate trajectories of a system, but also serves as a stepping-stone for many more results, including the proof of Theorem~\ref{th:3d F has unbounded sensitivity}.
In the second subsection, we establish that discrete time counterparts of continuous time systems  inherit the bounded sensitivity property from the underlying continuous time system, showing the soundness  of the concept.


\subsection{\bf Spreading a Systems} \label{sec:spread}
In this subsection, we study \emph{spreading} a dynamical system $F$, and show that if $F$ has bounded sensitivity, then its spread systems as well have bounded sensitivity in a weaker sense.
\begin{definition}[$\epsilon$-Spread System] \label{def:spread}
Consider a dynamical system $F(\cdot)$, and an $\epsilon\ge 0$. For every point $x$ in the domain, let
\begin{equation}
\tilde{F}_\epsilon(x) = \conv \big\{\drift \,\,\big|\,\, \drift\in F(y),\, y\in \ball_\epsilon(x)\big\},
\end{equation}
where $\ball_\epsilon(x)$ is the closed Euclidean ball of radius $\epsilon$ centered at $x$, and for a subset $\cal{S}$  of $\R^n$, $\conv(\cal{S})$ stands for the convex hull of $\cal{S}$. 
Then, we refer to $\tilde{F}_\epsilon$ as the $\epsilon$-spread of $F$. 
\end{definition}

The definition of a spread-system allows for the trajectories to follow the drift of  a neighbouring point.
Such models find applications in control systems, where the control applied is chosen on the basis of noisy state measurements. 
In this view, 
trajectories of $\tilde{F}_\epsilon$ can be perceived as \emph{approximate solutions} of $F$. 
Note that
given an initial point $p$, the unperturbed trajectories of $\tilde{F}_\epsilon$ that emanate from $p$ are not typically unique.

There are several notions of a \emph{generalized solution} of a differential equation, including weak solutions 
\cite{Evan88} 
and viscosity solutions \cite{Barl13}. 
These are the solutions that satisfy the differential equation almost everywhere, while allowing for non-differentiability at some zero-measure set of times. 
In contrast, a solution of a spread system of $F$ may satisfy the differential equation $\dot{x}\in F(x)$ at no point of time whatsoever.
In fact, the generalized solution of a differential equation are primarily developed to deal with non-differentiability of the solutions, while spread-systems allow for uncertainty about the current state, and lead to a notion of approximate solutions.

Despite the fact that several spread solutions can emerge form the same initial point, it turns that if a system has bounded sensitivity, then a weaker notion of bounded sensitivity still pertains to its spread systems.

\begin{theorem}[Sensitivity of Spread Systems] \label{th:spread sys}
Consider a dynamical system $F$ and an $\epsilon>0$.
Let $\ptraj(\cdot)$ be a perturbed trajectory of the $\epsilon$-spread system, $\tilde{F}_\epsilon$, of $F$, corresponding  to perturbation $\Prt(\cdot)$.
\begin{enumerate}[label={(\alph*)}, ref={\ref{th:spread sys} (\alph*)}]
\item
Suppose that \eqref{eq:bis} is valid with constant $\divconst$, and let $x(\cdot)$ be an unperturbed trajectory of $\tilde{F}_\epsilon$, initialized at $x(0)=\ptraj(0)$. 
Then, for any $t\ge 0$,
\begin{equation}\label{eq:bis of spread systems}
\Ltwo{\ptraj(t)- x(t)} \,\le \, \divconst\, \Big(2\epsilon+\sup_{\tau\le t} \Ltwo{ \Prt(\tau)}\Big) \, +\,3\epsilon.
\end{equation}

\item
Suppose that \eqref{eq:bpis} is valid with constant $\divconst$, and let $\ptraj'(\cdot)$ be a perturbed trajectory of $\tilde{F}_\epsilon$, corresponding  to perturbation $\Prt'(\cdot)$, and initialized at $\ptraj'(0)=\ptraj(0)$. 
Then, for any $t\ge 0$,
\begin{equation}\label{eq:bpis of spread systems}
\Ltwo{\ptraj(t)- \ptraj'(t)} \,\le \, \divconst\, \Big(2\epsilon+\sup_{\tau\le t} \Ltwo{ \Prt(\tau) - \Prt'(\tau)}\Big)\, +\,3\epsilon,
\end{equation}
\end{enumerate}
\end{theorem}
The proof is given in Section \ref{sec:proof th spread sys}, and involves  constructing a perturbation function,  the spread trajectory of which is a corresponding perturbed trajectory of the initial system.


We wish to point that a similar phenomenon has been previously studied in the literature of input-to-state stability \cite{CaiT13}.
Given an external disturbance $u(\cdot)$ and an $\epsilon>0$, it is shown in \cite{CaiT13} that if the dynamical system $\dot{x}\in F(x,u)$  is input-to-state stable,
then the system $\dot{x}\in \bigcup_{y\in B_\epsilon(x)} F(y,u)$ is also input-to-state stable, for sufficiently small values of  $\epsilon$.

\medskip



Leveraging Theorem~\ref{th:spread sys}, in the rest of this subsection we study a special type of transformation; and show that \emph{convolution with a kernel} preserves bounded sensitivity. 

\begin{definition}[Kernel and Convolution] \label{def:kernel}
For an $\epsilon >0$, an $\epsilon$-kernel is any integrable function $h:\ball_\epsilon(0)\to \R_+$.
Given a dynamical system $F$ and an $\epsilon$-kernel $h$, we define their convolution  $F*h$ as  
\begin{equation}
\big(F*h\big)\, (x) \,=\, \int_{\ball_\epsilon(x)} \drift_x(y) \,  h(x-y)\,dy,  \qquad \forall x\in\R^n,
\end{equation}
where $\drift_x(y) $ is an arbitrary vector in $F(y)$.
\end{definition}

\begin{corollary} \label{th:kernel}
Consider a dynamical system $F$ for which \eqref{eq:bis} is valid,  and let   $h(\cdot)$ be an $\epsilon$-kernel, for some  $\epsilon>0$.
Then, for any unperturbed trajectory $x(\cdot) $ of $F*h$, and any perturbed trajectory $\ptraj(\cdot)$ of  $F*h$ corresponding  to perturbation $\Prt(\cdot)$,
\begin{equation}\label{eq:bis of convolution}
\Ltwo{\ptraj(t)- x(t)} \,\le \, \divconst\, \Big(2\epsilon+\sup_{\tau\le t} \Ltwo{ \Prt(\tau)}\Big) \, +\,3\epsilon , \quad\forall t\in \R_+.
\end{equation}
\end{corollary}

\begin{proof}
Without loss of generality, assume that $\int_{\ball_\epsilon(0)} h(x)\,dx=1$.
Then, for any $x\in\R^n$,
\begin{equation}
\big(F*h\big)\, (x) \,\in\, \conv \big\{\drift \,\,\big|\,\, \drift\in F(y),\, y\in \ball_\epsilon(x)\big\} \,=\, \tilde{F}_\epsilon(x),
\end{equation}
where $\tilde{F}_\epsilon(x)$ is the $\epsilon$-spread of $F$.
Therefore, every perturbed (respectively, unperturbed) trajectory of $(F*h)$ is also a perturbed (unperturbed) trajectory of $\tilde{F}_\epsilon(x)$.
The corollary then follows from Theorem~\ref{th:spread sys}.
\end{proof}

A similar result is valid for sensitivity bounds of the form \eqref{eq:bpis}.

\medskip


\subsection{\bf Time Discretization } \label{subsec:discrete}
A discrete time trajectory is attained by taking (small)  steps along the drifts of a continuous time system. Formally,

\begin{definition}[Discrete Time Trajectories] \label{def:disc traj}
Consider a continuous time dynamical system $F:\R^n\to 2^{\R^n}$ and a function $\dpert:\Z_+\to\R^n$, which we refer to as  discrete time perturbation.
We then call $\dtraj(\cdot)$  a \emph{discrete time perturbed trajectory} corresponding to the perturbation function $V(\cdot)$, if there exists a function $\ddrift:\Z_+\to\R^n $ such that
\begin{equation}\label{eq:def of discrete traj}
\begin{split}
\dtraj(t+1) \, &=\, \sum_{k\le t} \ddrift(k)\, + \, \dpert(t),\quad \forall t\in \Z_+,\\
\ddrift(t) \, &\in\, F\big( \dtraj(t)\big),\quad \forall t\in \Z_+.
\end{split}
\end{equation}
\end{definition}

Discrete time trajectories correspond to systems that operate in slotted times. Examples include the queue lengths dynamics of job scheduling algorithms \cite{Neel10}. 
The following theorem  shows that in any  system whose continuous time trajectories have bounded sensitivity, a similar property also holds for its discrete time trajectories. 
\begin{theorem}[Sensitivity in Discrete Time] \label{th:main dis}
Consider a  dynamical system $F$, and let $\dtraj(\cdot)$ be a discrete time perturbed trajectory of $F$ corresponding to perturbation function $\dpert(\cdot)$. For every $k\in\Z_+$, let $\mu_k \triangleq \dtraj(k+1)-\dtraj(k) - \dpert(k)$, which is a vector in  $F\big(\dtraj(k)\big)$ (cf.~\eqref{eq:def of discrete traj}). 
\begin{enumerate}[label={(\alph*)}, ref={\ref{th:main dis} (\alph*)}]
\item \label{th:discrete pwc bound}
Suppose that a bound of type \eqref{eq:bis} is valid with constant $\divconst$. Let $x(\cdot)$ be the continuous time unperturbed trajectory of $F$ initialized at $x(0)=\dtraj(0)$. 
Then, for any $k\in\Z_+$,
\begin{equation}\label{eq:disc bis}
\Ltwo{x(k)-\dtraj(k)}\, \le\ \divconst\, \left(\max_{j< k} \,\ltwo{ \mu_j} \,+ \,\max_{j< k} \,\Ltwo{ \dpert(j)}\right).
\end{equation}

\item \label{th:discrete bpis}
Suppose that a bound of type \eqref{eq:bpis} is valid with constant $\divconst$.
Consider a discrete time  perturbation function $\dpert'(\cdot)$ and the corresponding discrete time perturbed trajectory $\dtraj'(\cdot)$  initialized at $\dtraj'(0)=\dtraj(0)$. 
For any $k\in\Z_+$, let  $\mu_k' \triangleq \dtraj'(k+1)-\dtraj'(k) - \dpert'(k)$. 
Then, for any $k\in\Z_+$,
\begin{equation} \label{eq:disc bpis}
\Ltwo{\dtraj'(k)-\dtraj(k)}\, \le\,  \divconst\, \left(\max_{j< k} \Ltwo{ \mu_j- \mu'_j} \,+ \,\max_{j< k} \Ltwo{ \dpert(i)- \dpert'(i)}\right).
\end{equation}
\end{enumerate}
\end{theorem}

The proof is given in Section \ref{sec:proof disc}. 
The high level idea is to simulate the  discrete time perturbed trajectories by continuous time perturbed trajectories, and take advantage of  the bounded sensitivity properties \eqref{eq:bis} and \eqref{eq:bpis}.


In Theorem~\ref{th:main dis}, unlike its continuous time counterparts, the deviation  bound also depends on the maximum jump size of the discrete time system, i.e., $\|\mu_k\|$. 
This dependency is inevitable because the distance between continuous time and discrete time trajectories cannot go arbitrarily small, even when the perturbation is zero. 

The next corollary is a consequence of Theorems~\ref{th:p1 sensitivity} and~\ref{th:main dis}.
\begin{corollary}
Consider an FPCS system, a continuous time unperturbed trajectory $x(\cdot)$ and a discrete time perturbed trajectory $\dtraj(\cdot)$ corresponding to perturbation $\dpert(\cdot)$. Then, for any $k\in\Z_+$,
\begin{equation}
\Ltwo{x(k)-\dtraj(k)}\, \le\,   \divconst\, \left(\mumax \,+ \,\max_{j< k} \Ltwo{ \sum_{i=0}^{j} \dpert(i)}\right),
\end{equation}
where  $C$ is the constant of Theorem~\ref{th:p1 sensitivity} and $\mumax$ is a constant independent of the trajectories.
\end{corollary}


\medskip

\section{\bf Proof of Theorem \ref{th:linear}} \label{sec:proof th linear}

We start with a well-known result on  solvability of linear dynamical systems.
\begin{lemma}[Solution of a Linear System]  \label{lem:completeness of linear}
Given a measurable perturbation function $U(\cdot)$, and an initial condition $x(0)=x_0$, the linear dynamical system $\dot{x}=Ax$ has a unique perturbed trajectory,  
of the following form 
\begin{equation}\label{eq:closed form formul for pert traj of linear sys}
x(t)  \,=\,  e^{At} x(0) \,+\, U(t) \,+\, A e^{At} \int_{0}^t e^{-A\tau}U(\tau)\,d\tau, \qquad \forall t\ge 0.
\end{equation}
\hide{
A perturbed linear dynamical system $\dot{x}(t)= Ax(t)+u(t)$, where $u(\cdot)$ is a measurable function, has a unique solution of the form
\begin{equation}\label{eq:closed form formul for pert traj of linear sys}
x(t)  \,=\,  e^{At} x(0) +  \int_{0}^t e^{A\tau}u(t-\tau)\,d\tau. 
\end{equation}
}
\end{lemma}
For completeness, we give the proof in Appendix~\ref{app:proof sol linear}. 
\oli{I believe a similar result must be there, in the literature. 
If $U$ is differentiable, an equivalent result is well-known (see e.g., the Gronwall's inequality, or the wikipedia page on linear ODE). 
If $U$ is not differentiable, \cite{AmesP97} gives similar formulas, but for the one dimensional case (where $x(t)$ and $U(t)$ are scalars).}

Consider a pair $U_1(\cdot)$ and $U_2(\cdot)$ of perturbation functions and a pair of corresponding perturbed trajectories $\xt_1(\cdot)$  and $\xt_2(\cdot)$, with the same initial condition $\xt_1(0) = \xt_2(0)$. 
Then, Lemma \ref{lem:completeness of linear} implies that for any $t\ge0$,
\begin{equation} \label{eq:ltwo ptraj-x bounded by integral eA pert}
\begin{split}
\Ltwo{\ptraj_1(t)-\ptraj_2(t)} \, &=\, \Ltwo{U_1(t)-U_2(t) \,+\, A  \int_{0}^t e^{A\tau} \Big(U_1(t-\tau)-U_2(t-\tau)\Big) \,d\tau}\\
&\le\, \Ltwo{U_1(t)-U_2(t)} \,+\,  \Ltwo{  \int_{0}^t A e^{A\tau} \big(U_1(t-\tau)-U_2(t-\tau)\big) \,d\tau}.
\end{split}
\end{equation}

Consider the Jordan normal form of $A$,
\begin{equation} \label{eq:jordan joon}
A \,=\, P D P^{-1},  \qquad  D = \Lambda + B,
\end{equation}
where $\Lambda$ is the diagonal matrix of eigenvalues of $A$, and  $B$ is a matrix 
with some superdiagonal entries equal to one, and all other entries equal to zero.  
It follows that
\begin{equation}\label{eq:formula of e to the A}
e^{At} \,=\, Pe^{Dt}P^{-1}.
\end{equation}
Fix a $t>0$, and let
\begin{equation}
\maxpert \triangleq \sup_{\tau\le t} \Ltwo{U_1(t-\tau)-U_2(t-\tau)}.
\end{equation}
For any $\tau\in [0,t]$, let 
\begin{equation}
V(\tau) \triangleq P^{-1} \Big(  U_1(t-\tau)-U_2(t-\tau)   \Big).
\end{equation}
Then, for any $\tau\in[0,t]$,
\begin{equation}\label{eq:V tau bound by sup u}
\begin{split}
\Ltwo{V(\tau)} & \,\le\, \frac{\maxpert}{\sigmin},
\end{split}
\end{equation}
where $\sigmin$ is the smallest singular value of $P$. 
Since $P$ is invertible, $\sigmin>0$. 
It follows from (\ref{eq:ltwo ptraj-x bounded by integral eA pert}) and (\ref{eq:formula of e to the A}) that 
\begin{equation}\label{eq:linear d ptraj sigmax int eD 1} 
\begin{split}
\Ltwo{\ptraj_1(t)-\ptraj_2(t)} &\le\, \Ltwo{U_1(t)-U_2(t)} \,+\,  \Ltwo{  \int_{0}^t A e^{A\tau} \Big(U_1(t-\tau)-U_2(t-\tau)\Big) \,d\tau}\\
&=\, \Ltwo{U_1(t)-U_2(t)} \,+\,  \Ltwo{ \int_{0}^t  P D P^{-1} \,  P e^{D \tau} P^{-1} \, \Big(U_1(t-\tau)-U_2(t-\tau)\Big) \,d\tau}\\
&=\, \Ltwo{U_1(t)-U_2(t)} \,+\,  \Ltwo{P \int_{0}^t  D e^{D \tau} V(\tau) \,d\tau}.
\end{split}
\end{equation}

From the SOF assumption and Lemma~\ref{lem:SOF linear}, in the Jordan decomposition, every block associated with the zero eigenvalue has unite size. Equivalently, $D$ has the following form
\begin{equation}
D = \left[ \begin{array}{c|c}   0 & 0 \\ \hline 0 &G    \end{array} \right],
\end{equation}
where $G$ comprises the blocks of $D$ corresponding to non-zero eigenvalues. 
Then, for any $\tau\ge 0$,  
\begin{equation}
e^{D\tau} =  \left[ \begin{array}{c|c}   I & 0 \\ \hline 0 &e^{G\tau}    \end{array} \right].
\end{equation}
For any $\tau\ge0$, consider the decomposition
\begin{equation}
V(\tau) = \left[\begin{array}{c}  V_0(\tau) \\ V_1(\tau) \end{array}   \right],
\end{equation}
where $V_0$ is a vector of length equal to the multiplicity of the zero eigenvalue in $A$.
Therefore, for any $\tau\ge 0$,
\begin{equation} \label{eq:norm D tilde equals norm D}
\Ltwo{D e^{D \tau} V(\tau)} \,=\,  \big\| \Bigg[ \begin{array}{c} 0\\ \hline G e^{G \tau} V_1(\tau) \end{array} \Bigg]  \big\|  \,=\,  \Ltwo{G e^{G \tau} V_1(\tau)}.
\end{equation}

For a matrix $M$, we denote its  Frobenius norm by  $\|M\|_F$. We also let  $\sigmax$ be the largest singular value of $P$. 
It follows from \eqref{eq:linear d ptraj sigmax int eD 1}  that
\begin{equation}\label{eq:linear d ptraj sigmax int eD 2} 
\begin{split}
\Ltwo{\ptraj_1(t)-\ptraj_2(t)} \, &\le\, \Ltwo{U_1(t)-U_2(t)} \,+\,  \Ltwo{P \int_{0}^t  D e^{D \tau} V(\tau) \,d\tau}\\
&\le \, \Ltwo{U_1(t)-U_2(t)} \,+\,  \sigmax\int_{0}^t  \Ltwo{D e^{D \tau} V(\tau)} \,d\tau\\
&=\, \Ltwo{U_1(t)-U_2(t)} \,+\,  \sigmax   \int_{0}^t  \Ltwo{G e^{G \tau} V_1(\tau)} \,d\tau\\
&\le\, \theta \,+\,  \sigmax \int_{0}^t \Ltwo{G e^{G \tau}}_F\, \|{ V_1(\tau) }\|\,d\tau \\ 
&\le\, \left( 1\,+\,  \frac{ \sigmax }{ \sigmin }\,  \int_{0}^\infty \Ltwo{G e^{G \tau}}_F\,d\tau \right) \,\theta,
\end{split}
\end{equation}
where the equality is from \eqref{eq:norm D tilde equals norm D} and the last inequality is due to \eqref{eq:V tau bound by sup u}. 
Since all eigenvalues of $G$ have negative real parts, the integral in the right hand side of \eqref{eq:linear d ptraj sigmax int eD 2}  is finite. 
Then, bounded sensitivity of the SOF linear system follows from  \eqref{eq:linear d ptraj sigmax int eD 2}.

\hide{
By the SOF assumption, for the eigenvalue $\lambda=0$ (if any), the rank of eigenspace of $\lambda$ is equal to the multiplicity of $\lambda$. Hence, in the Jordan decomposition, every block associated with $\lambda$ has size one, i.e., we can represent $D$ as
\begin{equation}
D = \left[ \begin{array}{c|c}   0 & 0 \\ \hline 0 &\tilde{D}    \end{array} \right],
\end{equation}
where $\tilde{D}$ comprises the blocks of $D$ corresponding to non-zero eigenvalues. Hence, for any $\tau$,  we can write
\begin{equation}
e^{D\tau} =  \left[ \begin{array}{c|c}   I & 0 \\ \hline 0 &e^{\tilde{D}\tau}    \end{array} \right],\qquad v(\tau) = \left[\begin{array}{c}  v_0(\tau) \\ \tilde{v}(\tau) \end{array}   \right].
\end{equation}
Therefore,
\begin{equation}
\begin{split}
\Ltwo{ \int_{0}^t e^{D \tau} v(\tau)\,d\tau }&\, \le\, \Ltwo{ \int_{0}^t  v_0(\tau)\,d\tau } + \Ltwo{ \int_{0}^t e^{\tilde{D} \tau}  \tilde{v}(\tau)\,d\tau } \\
&\, \le\, \Ltwo{V(\tau)} + \Ltwo{ \int_{0}^t e^{\tilde{D} \tau}  \tilde{v}(\tau)\,d\tau }.
\end{split}
\end{equation}
Thus, in the light of (\ref{eq:V tau bound by sup u}) we can forget about the zero eigenvalues and present a bound for $ \Ltwo{ \int_{0}^t e^{\tilde{D} \tau}  \tilde{v}(\tau)\,d\tau }$. So, without loss of generality, in the sequel we assume that $A$ does not have zero eigenvalues. In this case, by SOF assumption, every eigenvalue of $A$ has  negative real part.

In (\ref{eq:jordan joon}), $B$ is nilpotent with $B^r=0$, where $r$ is the largest multiplicity of eigenvalues of $A$. Moreover, $B$ commutes with $\Lambda$, and as a result, $e^{\Lambda+B}=e^{B}e^{\Lambda}$. It follows that
\begin{equation}\label{eq:formula of e to the A 2}
e^{Dt} \,=\, e^{(\Lambda+B)t} \,=\, e^{Bt}e^{\Lambda t}\,=\, \sum_{k=0}^{r-1} \frac{1}{k!} B^{k} \big(t^k e^{\Lambda t}\big).
\end{equation}
 Since $B$ is element-wise less than or equal to a permutation matrix, $\sigmax(B)\le 1$. Hence, for any $k\ge 0$, $\sigmax\big(B^k\big)\le 1$. It follows from (\ref{eq:linear d ptraj sigmax int eD}) and (\ref{eq:formula of e to the A 2}) that,
\begin{equation}\label{eq:int e to the A pert le sum k ints}
\begin{split}
\Ltwo{\ptraj_1(t)-\ptraj_2(t)} & \le \sigmax \,\Ltwo{ \int_{0}^t e^{D \tau} v(\tau)\,d\tau }\\
&=   \sigmax \Ltwo{\sum_{k=0}^{r-1}\frac{1}{k!}   B^{k} \int_{0}^t \tau^k e^{\Lambda \tau} v(\tau)\,d\tau }\\
& \le   \sigmax \sum_{k=0}^{r-1}\frac{1}{k!}\,\sigmax(B^{k}) \,\Ltwo{  \int_{0}^t \tau^k e^{\Lambda \tau} v(\tau)\,d\tau }\\
&\le \sigmax  \sum_{k=0}^{r-1}\frac{1}{k!} \Ltwo{ \int_{0}^t \tau^k e^{\Lambda \tau} v(\tau)\,d\tau },
\end{split}
\end{equation}
\hide{
\begin{equation}\label{eq:int e to the A pert le sum k ints}
\begin{split}
\Ltwo{\ptraj_1(t)-\ptraj_2(t)} & \le \sigmax \,\Ltwo{ \int_{0}^t e^{D \tau} v(\tau)\,d\tau }\\
&\Ltwo{\int_0^t e^{A\tau}\Big(u_1(t-\tau)-u_2(t-\tau)\Big) \,d\tau} \\
&\qquad= \Ltwo{\sum_{k=0}^{r-1} \frac{1}{k!} P B^{k} \int_{0}^t \tau^k e^{\Lambda \tau} P^{-1}\Big(u_1(t-\tau)-u_2(t-\tau)\Big)\,d\tau }\\
&\qquad =   \Ltwo{\sum_{k=0}^{r-1}\frac{1}{k!}   P B^{k} \int_{0}^t \tau^k e^{\Lambda \tau} v(\tau)\,d\tau }\\
&\qquad \le \sum_{i\in\textrm{blocks}}\, \Ltwo{ \sum_{k=0}^{r_i-1}\frac{1}{k!}   P_i B_i^{k} \int_{0}^t \tau^k e^{\Lambda_i \tau} v_i(\tau)\,d\tau }\\
&\qquad \le \sum_{i\in\textrm{blocks}} \, \sum_{k=0}^{r_i-1}\frac{1}{k!} \Ltwo{  P_i B_i^{k} \int_{0}^t \tau^k e^{\Lambda_i \tau} v_i(\tau)\,d\tau }\\
&\qquad \le \sigmax  \sum_{i\in\textrm{blocks}} \, \sum_{k=0}^{r_i-1}\frac{1}{k!} \Ltwo{ \int_{0}^t \tau^k e^{\Lambda_i \tau} v_i(\tau)\,d\tau },
\end{split}
\end{equation}
}
\hide{
\begin{equation}\label{eq:int e to the A pert le sum k ints}
\begin{split}
&\Ltwo{\int_0^t e^{A\tau}\Big(u_1(t-\tau)-u_2(t-\tau)\Big) \,d\tau} \\
&\qquad= \Ltwo{\sum_{k=0}^{r-1} \frac{1}{k!} P B^{k} \int_{0}^t \tau^k e^{\Lambda \tau} P^{-1}\Big(u_1(t-\tau)-u_2(t-\tau)\Big)\,d\tau }\\
&\qquad =   \Ltwo{\sum_{k=0}^{r-1}\frac{1}{k!}   P B^{k} \int_{0}^t \tau^k e^{\Lambda \tau} v(\tau)\,d\tau }\\
&\qquad \le   \sum_{k=0}^{r-1}\frac{1}{k!} \Ltwo{  P B^{k} \int_{0}^t \tau^k e^{\Lambda \tau} v(\tau)\,d\tau }\\
&\qquad\le \sigmax  \sum_{k=0}^{r-1}\frac{1}{k!} \Ltwo{ \int_{0}^t \tau^k e^{\Lambda \tau} v(\tau)\,d\tau },
\end{split}
\end{equation}
}
For every $k\in\Z_+$ and every $\tau\ge 0$, let
\begin{equation}
f_k(\tau)\triangleq \tau^k e^{\Lambda \tau},
\end{equation}
and
\begin{equation}\label{eq:divconst k le infty}
\begin{split}
\divconst_k &\triangleq \int_0^\infty \Ltwo{f'_{k}(\tau)}\,d\tau\\
& = \int_0^\infty \Ltwo{\frac{d}{dt}\big(\tau^ke^{\Lambda\tau}\big)}\,d\tau\\
& \le \sum_{i=1}^n \int_0^\infty \Ltwo{\frac{d}{dt}\big(\tau^ke^{\lambda_i\tau}\big)}\,d\tau\\
& = \sum_{i=1}^n \int_0^\infty \Ltwo{\big(k\tau^{k-1}+\lambda_i \tau^k\big)e^{\lambda_i\tau}}\,d\tau\\
& \le \sum_{i=1}^n \int_0^\infty \Big(k\tau^{k-1}+\Ltwo{\lambda_i} \tau^k\Big)e^{\mathrm{real}(\lambda_i)\tau}\,d\tau\\
& < \infty,
\end{split}
\end{equation}
\hide{
\begin{equation}\label{eq:divconst k le infty}
\begin{split}
\divconst_k &\triangleq \max_i \int_0^\infty \Ltwo{f'_{k,i}(\tau)}\,d\tau\\
& = \int_0^\infty \Ltwo{\frac{d}{dt}\big(\tau^ke^{\lambda_i\tau}\big)}\,d\tau\\
& =\int_0^\infty \Ltwo{\big(k\tau^{k-1}+\lambda_i \tau^k\big)e^{\lambda_i\tau}}\,d\tau\\
& \le \int_0^\infty \Big(k\tau^{k-1}+\Ltwo{\lambda_i} \tau^k\Big)e^{\mathrm{real}(\lambda_i)\tau}\,d\tau\\
& < \infty,
\end{split}
\end{equation}
}
where the last inequality is due to the assumption $\mathrm{real}(\lambda_i)<0$.
Now, since $f_k(\cdot)$ is continuously differentiable and $v(\cdot)$ is Lebesgue integrable, it follows from the integration by parts (see e.g., Theorem 12.5 in \cite{Gord94}) that,
\begin{equation}\label{eq:int by parts}
\begin{split}
\int_{0}^t f_k(\tau) v(\tau)\,d\tau & = f_k(t) V(t) - f_k(0) V(0) - \int_{0}^t f'_k(\tau) V(\tau)\,d\tau\\
& = - f_k(0) V(0) - \int_{0}^t f'_k(\tau) V(\tau)\,d\tau,
\end{split}
\end{equation}
where the last equality is because $V(t)=\int_t^t P^{-1}u(t-\tau)\,d\tau= 0$. 
Plugging (\ref{eq:int by parts}) into (\ref{eq:int e to the A pert le sum k ints}), we have
\begin{equation} \label{eq:the big bound in the proof of th linear}
\begin{split}
\Ltwo{\ptraj_1(t)-\ptraj_2(t)} 
& \le  \sigmax \sum_{k=0}^{r-1}\frac{1}{k!}\, \Ltwo{\int_{0}^t f_k(\tau) v(\tau)\,d\tau}\\
& \le  \sigmax \sum_{k=0}^{r-1}\frac{1}{k!}\Ltwo{f_k(0) V(0)}  +  \sigmax \sum_{k=0}^{r-1}\frac{1}{k!} \Ltwo{\int_{0}^t f'_k(\tau) V(\tau)\,d\tau}\\
& = \sigmax\Ltwo{V(0)} + \sigmax \sum_{k=0}^{r-1}\frac{1}{k!} \Ltwo{\int_{0}^t f'_k(\tau) V(\tau)\,d\tau}\\
& \le \sigmax\Ltwo{V(0)} + \sigmax \sum_{k=0}^{r-1}\frac{1}{k!} \int_{0}^t \Ltwo{f'_k(\tau)}\,\Ltwo{V(\tau)} \,d\tau\\
& \le \frac{2\maxpert \sigmax}{\sigmin} + \frac{2\maxpert \sigmax}{\sigmin} \sum_{k=0}^{r-1}\frac{1}{k!} \int_{0}^t \Ltwo{f'_k(\tau)} \,d\tau\\
& \le \frac{2 \sigmax}{\sigmin} \left(1+\sum_{k=0}^{r-1}\frac{\divconst_k}{k!} \right)\maxpert.
\end{split}
\end{equation}
where the third equality is because $f_k(0)=0$, for $k\ge 1$, and the last two inequalities are due to (\ref{eq:V tau bound by sup u}) and (\ref{eq:divconst k le infty}), respectively. 
This establishes the first part of the theorem.
}

For the second part, if the linear system is not SOF, then it is either unstable or has a periodic orbit. If the system is unstable, then a small perturbation at time zero can cause a perturbed trajectory $\ptraj_1(\cdot)$ with initial condition $\ptraj_1(0)=0$ to have $\lim_{t\to\infty} \Ltwo{\ptraj_1(t)}=\infty$. 
On the other hand, $x(t)=0$, for all $t\ge0$, is an unperturbed trajectory. Then, the distance between $\ptraj_1(\cdot)$ and the unperturbed trajectory $x(\cdot)$ grows unbounded.
In the second case, if the system is not orbit-free, consider a periodic orbit $x(t)=e^{At}x_0$, with $x(t_0)=x(0)=x_0$, for some $t_0>0$. 
Then, $u(t)=e^{At}x_0$ has a bounded integral. However, $\ptraj(t)=(t+1)e^{At}x_0$ satisfies the differential equation $\frac{d}{dt}\ptraj(t)=A\ptraj(t)+u(t)$, and hence is a perturbed trajectory whose deviation from $x(t)$ is unbounded as $t$ goes to infinity. 
We conclude that, in either case, if the linear system is not SOF, a bounded perturbation can cause unbounded changes in the trajectories, and there would be no constant $\divconst$ to satisfy (\ref{eq:bpis}).

For the third part, if $A$ is diagonalizable, then $G$ is a diagonal matrix with the non-zero eigenvalues of $A$ on its main diagonal.
Therefore,
\begin{equation} \label{eq:d tilde integral diagonalizable case}
\begin{split}
\int_{0}^\infty \Ltwo{G e^{G\tau}}\,d\tau \, &\le \,  \int_{0}^\infty \sum_{i=1}^n  |\lambda_i |  e^{\real (\lambda_i)\,\tau} \,d\tau 
\, = \, \sum_{i=1}^n\frac{|{\lambda_i}|}{\big\lvert  \real(\lambda_i)  \big\rvert}.
\end{split}
\end{equation}
Plugging \eqref{eq:d tilde integral diagonalizable case} into \eqref{eq:linear d ptraj sigmax int eD 2}  implies \eqref{eq:constant C for linear in diagonalizable case}.

Finally, if $A$ is symmetric, then all eigenvalues are real and $P$ is orthonormal. 
Therefore, $|{\lambda_i}|/{\lvert  \real(\lambda_i) \rvert}=1$ and $\sigmax = \sigmin=1$. 
Then, \eqref{eq:constant C for linear in diagonalizable case} can be further simplified into $\Ltwo{\ptraj_1(t)-\ptraj_2(t)} \le  \big(n+1\big)\maxpert$.

\medskip


\section{\bf Proof of Theorem \ref{th:spread sys}}  \label{sec:proof th spread sys} 
We start with two lemmas, and then give the proof of Theorem \ref{th:spread sys}.
\begin{lemma} \label{lem:simul spread}
Consider a dynamical system $F$ and its $\epsilon$-spread system $\tilde{F}_\epsilon$, for some $\epsilon>0$. 
\begin{itemize}
\item[a)]
Let $\ptraj(\cdot)$ be a perturbed trajectory of $\tilde{F}_\epsilon$, corresponding to perturbation $\pert(\cdot)$.
Then, for any $\delta>0$, there exists a perturbation $\pp(\cdot)$, and a corresponding perturbed trajectory $\yt(\cdot)$ of $F$, 
such that for any $t\ge 0$,
\begin{equation}\label{eq:yt near ptraj}
\Ltwo{\yt(t) - \ptraj(t) }\, \le \, \epsilon+\delta,
\end{equation}
\begin{equation} \label{eq:pp close to pert} 
\sup_{\tau \le t} \Ltwo{ \pp(\tau)} \,\le \, \sup_{\tau \le t}\Ltwo{ \pert(\tau)} \,+\, \epsilon \,+\, \delta.
\end{equation}

\item[b)]
For any pair $\pert_1(\cdot)$  and $\pert_2(\cdot)$ of perturbations and corresponding pair $\ptraj_1(\cdot)$ and $\ptraj_2(\cdot)$ of perturbed trajectories of $\tilde{F}_\epsilon$, and for any $\delta>0$, there exists a pair $\pert'_1(\cdot)$  and $\pert'_2(\cdot)$ of perturbations and corresponding pair $\yt_1(\cdot)$ and $\yt_2(\cdot)$ of perturbed trajectories of  $F$ such that
\begin{equation}\label{eq:yt near ptraj 2}
\Ltwo{\yt_i(t) - \ptraj_i(t) }\, \le \, \epsilon+\delta, \,\qquad i=1,2,
\end{equation}
\begin{equation} \label{eq:pp close to pert 2} 
\sup_{\tau \le t}\Ltwo{ \pert'_1 -  \pert'_2(\tau)} \,\le \, \sup_{\tau \le t}\Ltwo{ \pert_1(\tau) - \pert_2(\tau)} \,+\, 2\epsilon \,+\, \delta.
\end{equation}
\end{itemize}
\end{lemma}

The proof is elaborate and is given in Appendix~\ref{app:proof spread}.
The convex hull, $\conv(\cdot)$ in the definition of $\tilde{F}_\epsilon$ brings a tremendous  amount of complication to the proof of Lemma~\ref{lem:simul spread}. 
Here, to see the main idea, we present and prove a simpler counterpart of Lemma~\ref{lem:simul spread}.

\begin{lemma}
Consider a dynamical system $F$ and an $\epsilon>0$.
For any $x\in\R^n$, let 
\begin{equation}
\hat{F}(x) = \big\{\drift \,\big|\, \drift\in F(y),\, y\in \ball_\epsilon(x)\big\}.
\end{equation}
Let $\ptraj(\cdot)$ be a perturbed trajectory of $\hat{F}$, corresponding to some perturbation $\pert(\cdot)$.
Then,  there exist a perturbation $\pp(\cdot)$, and a corresponding perturbed trajectory $\yt(\cdot)$ of $F$ 
that satisfy \eqref{eq:yt near ptraj} and \eqref{eq:pp close to pert} for $\delta=0$.
\end{lemma}

\begin{proof}
By definition, 
\begin{equation}
\ptraj(t) = \ptraj(0) \,+\, \int_0^t {\xi}( \tau) \,d\tau \,+\, U(t),\quad \forall t\ge0,
\end{equation}
where ${\xi}( \tau)\in \hat{F}\big(\ptraj(\tau)\big)$. Therefore, for any $\tau$, there exists a $\yt(\tau)$ in the $\epsilon$-neighbourhood of $\ptraj(\tau)$, such that ${\xi}(\tau)\in {F}\big(\yt(\tau)\big)$. Then, for any $t$,
\begin{equation}
\yt(t) \,=\, \ptraj(t) + \Big(\yt(t) - \ptraj(t)\Big) \, =\, \ptraj(0) \,+\, \int_0^t {\xi}( \tau) \,d\tau \,+\, \Big(U(t) +\yt(t) -\ptraj(t)  \Big).
\end{equation}
Hence, $\yt(\cdot)$  is a perturbed trajectory of $F$ associated with perturbation $\pp(t) =\pert(t) + \yt(t) - \ptraj(t)$. 
Since for any $t\ge0$, $\Ltwo{\yt(t)- \ptraj(t)}\le \epsilon$, then $\Ltwo{\pp(t)}\le \Ltwo{\pert(t)} + \epsilon$, and the lemma follows.

\end{proof}

\begin{proof}[\bf Proof of Theorem~\ref{th:spread sys}]

For Part (a), let 
\begin{equation}\label{eq:def of delta in proof of th3}
\delta=\frac\epsilon{2\big(\divconst+1\big)}.
\end{equation} 
From Lemma~\ref{lem:simul spread}, there exists a pair $\pp_1(\cdot)$ and $\pp_2(\cdot)$ of perturbation functions and a corresponding pair $\yt_1(\cdot)$ and $\yt_2(\cdot)$ of perturbed trajectories of $F$ such that for any $t\ge 0$,
\begin{equation} \label{eq:proof of th spread choice of aux trajs 1}
\begin{split}
\Ltwo{\yt_1(t) - \ptraj(t) }\, &\le \, \epsilon +\delta, \\
\Ltwo{\yt_2(t) - x(t) }\, &\le \, \epsilon +\delta,\\
\Ltwo{\sup_{\tau \le t} \pp_1(\tau)} \,&\le \, \Ltwo{\sup_{\tau \le t} \pert(\tau)} \,+\, \epsilon + \delta\\
\Ltwo{\sup_{\tau \le t} \pp_2(\tau)} \,&\le \, \epsilon + \delta.
\end{split}
\end{equation}
Let $y(\cdot)$ be an unperturbed trajectory of $F$ initialized at $\tilde{y}(0)=x(0)$. 
Then, 
\begin{equation}
\begin{split}
\Ltwo{\ptraj(t)- x(t)} &\,\le \, \Ltwo{\ptraj(t)- \yt_2(t)} + \Ltwo{\yt_2(t) - y(t)} + \Ltwo{y(t)- \yt_1(t)} + \Ltwo{\yt_1(t)-x(t)}\\
&\,\le \, (\epsilon+\delta) +  \Ltwo{\yt_2(t) - y(t)} + \Ltwo{y(t)- \yt_1(t)} + (\epsilon+\delta)\\
&\,\le \,  \divconst\, \sup_{\tau\le t} \Ltwo{ \pp_2(\tau)}
\,+\, \divconst\, \sup_{\tau\le t} \Ltwo{ \pp_1(\tau)} \,+\, 2\epsilon +2\delta \\
&\,\le \,  \divconst\, \Big(\sup_{\tau\le t} \Ltwo{ \pert(\tau)} + \epsilon+ \delta\Big)
\,+\, \divconst(\epsilon+\delta) \,+\, 2\epsilon +2\delta \\
& \,= \,\, \divconst\, \Big(2\epsilon+\sup_{\tau\le t} \Ltwo{ \Prt(\tau)}\Big) \,+\, 3\epsilon,
\end{split}
\end{equation}
where the equations are due to the triangle inequality, \eqref{eq:proof of th spread choice of aux trajs 1}, \eqref{eq:bis}, again \eqref{eq:proof of th spread choice of aux trajs 1}, and  \eqref{eq:def of delta in proof of th3}, respectively.
This completes the proof of Part (a).

\medskip

The proof of Part (b) is similar to the proof of Part (a). 
In view of Lemma~\ref{lem:simul spread}~(b), consider a pair $\pp_1(\cdot)$ and $\pp_2(\cdot)$ of perturbations and a corresponding pair $\yt_1(\cdot)$ and $\yt_2(\cdot)$ of perturbed trajectories of $F$ such that  for any $t\ge0$,
\begin{equation} \label{eq:proof of th spread choice of aux trajs 2}
\begin{split}
\Ltwo{\yt_i(t) - \ptraj_i(t) }\, &\le \, \epsilon +\delta, \qquad i=1,2,\\
\sup_{\tau \le t} \Ltwo{\pp_1(\tau) - \pp_2(\tau)} \,&\le \, \sup_{\tau \le t} \Ltwo{\pert_1(\tau)- \pert_2(\tau)} \,+\, 2\epsilon + \delta.
\end{split}
\end{equation}
Then, 
\begin{equation}
\begin{split}
\Ltwo{\ptraj(t)- x(t)} &\,\le \, \Ltwo{\ptraj(t)- \yt_2(t)} + \Ltwo{\yt_2(t) - \yt_1(t)} + \Ltwo{\yt_1(t)-\ptraj_1(t)}\\
&\,\le \,   \Ltwo{\yt_2(t) - \yt_1(t)}  + 2(\epsilon+\delta)\\
&\,\le \,  \divconst\, \sup_{\tau\le t} \Ltwo{ \pp_2(\tau) - \pp_1(\tau)} \,+\, 2\epsilon +2\delta \\
&\,\le \,  \divconst\, \Big(\sup_{\tau\le t} \Ltwo{ \pert_2(\tau) - \pert_1(\tau)} + 2\epsilon+\delta \Big)  \,+\, 2\epsilon +2\delta \\
& \,\le \,\, \divconst\, \Big(2\epsilon+\sup_{\tau\le t} \Ltwo{ \Prt_2(\tau)-\pert_1(\tau)}\Big) \,+\, 3\epsilon,
\end{split}
\end{equation}
where the relations are due to the triangle inequality, \eqref{eq:proof of th spread choice of aux trajs 2}, \eqref{eq:bpis}, again \eqref{eq:proof of th spread choice of aux trajs 2}, and  \eqref{eq:def of delta in proof of th3}, respectively.
This completes the proof of Theorem~\ref{th:spread sys}.
\end{proof}


\medskip

\section{\bf Proof of Theorem~\ref{th:3d F has unbounded sensitivity}} \label{app:proof unbounded sensitivity of F}

Here, we prove Theorem~\ref{th:3d F has unbounded sensitivity}
by first showing that the spread systems of $F$ (see Definition~\ref{def:spread}) admit spiral trajectories of the form depicted in Fig.~\ref{fig:3d example}~(b). 
We then  conclude that  spread systems of $F$ has unbounded sensitivity, in the sense that no constant $C$ can satisfy \eqref{eq:bis of spread systems}.
Finally, we use Theorem~\ref{th:spread sys} to show that $F$ has unbounded sensitivity.

To facilitate the presentation, we will be working with cylindrical coordinates and \emph{local cylindrical coordinates}.  
We represent a point $x(t)$ on a trajectory by $x(t) =\big(r,\phi, z\big)$ in cylindrical coordinates.
To represent the derivative of the trajectory (or speed vector) at point $x(t)$ we use {\bf local} cylindrical coordinates at $x(t)$,  that is $\dot{x}(t) = \alpha\rcor+\beta\phicor+\gamma\zcor$. Here $\alpha$ is the radial speed of the trajectory, $\beta$ is the trajectory speed along the direction of rotation around the $z$ axis at $x(t)$, and $\gamma$ is the trajectory speed along the $z$ axis. 
See Fig.~\ref{fig:cyl coor} for an illustration. 

\begin{figure} 
	\begin{center}
		{
			\includegraphics[width = .35\textwidth]{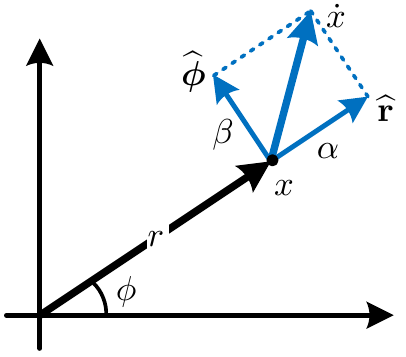}}
	\end{center}
	\caption{A two dimensional view of cylindrical and local cylindrical coordinates on the $z=0$ plane. The figure depicts a point $x$ and a vector $\dot{x}$  that emanates from $x$. 
	We write $x =\big(r,\phi, 0\big)$  to represent $x$ in cylindrical coordinates. We write $\dot{x} = \alpha\rcor+\beta\phicor+0\zcor$ to represent $\dot{x}$ in the local cylindrical coordinates at $x$.}
	\label{fig:cyl coor}
\end{figure}


Fix $\epsilon>0$  let $\tilde{F}_\epsilon$ be the $\epsilon$-spread system of $F$.

\begin{lemma}\label{claim:spread books}
For any $\epsilon>0$, there exists a constant $z_\epsilon<0$ and a smooth function $\alpha:\Omega\to\R_+$,  such that for any $r<1/4$, any $z<z_\epsilon$, and point $p= (r , z , z)$ in cylindrical coordinates,
\begin{equation}
\frac{-2r}{1 - z}\rcor \,+\,  r\phicor \,+\, \zcor\,\in\, \frac{1}{\alpha(p)}\tilde{F}_\epsilon(p).
\end{equation}
\end{lemma}
\begin{proof}
Fix a sufficiently large $n\in\N$, to be determined later. 
From Lemma~\ref{lem:properties of f}~(b), consider the level-set $f(r,\phi,z)=2\pi n$: 
\begin{equation}
z\big(r,\phi\big) \,=\, -2\pi n   \,+\, \frac1{2\pi n} \ln\Big(\cosh \big(2\pi n r\sin(\phi)\big)\Big)  \,+\,\frac{r^2}{1+2\pi n}.
\end{equation}
This level-set has the following representation in the Cartesian coordinates:
\begin{equation}\label{eq:eq of z in x y}
z\big(x,y\big) \,=\, -2\pi n   \,+\, \frac1{2\pi n} \ln\Big(\cosh \big(2\pi n \,y\big)\Big)  \,+\,\frac{x^2+y^2}{1+2\pi n}.
\end{equation}
Fix some $x_0,y_0\le 1/4$, and consider the vector
\begin{equation}\label{eq:scaled normal vec}
v\, =\, -\left(\left.\frac{dz}{dx} \right\rvert_{{(x_0,y_0)}}\right)\xcor \, - \, \left(\left.\frac{dz}{dy} \right\rvert_{{(x_0,y_0)}} \right) \ycor \,+\, \zcor.
\end{equation}
Then\footnote{Note that for any surface  $z(x,y)$, \eqref{eq:scaled normal vec} gives an orthogonal vector to that surface.}, 
$v$ is orthogonal to the surface 
$z(x,y)$ at point $\big(x_0,y_0,z(x_0,y_0)\big)$. 
Equivalently, $v$ is a scaling of the normal vector for level-set  $f(r,\phi,z)=2\pi n$ 
 at $\big(x_0,y_0,z(x_0,y_0)\big)$.   
Then, there exists a constant $\alphah(x,y,z)>0$ such that 
\begin{equation} \label{eq:grad x0y0z =av}
\alphah \big(x_0,y_0,z(x_0,y_0)\big)\, v \,=\, \grad f\big(x_0,y_0,z(x_0,y_0)\big).
\end{equation}
Note that the $z$ coordinate entry of $v$ is unity and, form Lemma~\ref{lem:properties of f}~(c), $\grad f$ is smooth. Then, the function $\alphah$ is smooth, as well.

We proceed by elaborating on the partial derivatives in \eqref{eq:scaled normal vec}. 
We have
\begin{equation}\label{eq:dzdx small}
-\left.\frac{dz}{dx} \right\rvert_{{(x_0,y_0)}}  = -\frac{2x_0}{1+2\pi n},
\end{equation}
and 
\begin{equation} \label{eq:partial derivative of z wrt y}
\begin{split}
-\left.\frac{dz}{dy} \right\rvert_{{(x_0,y_0)}} \,
=\, -\tanh\big(2\pi n\,y_0\big) \,-\, \frac{2 y_0}{1+2\pi n}.
\end{split}
\end{equation}
Let $g(y_0) = -\left.\frac{dz}{dy} \right\rvert_{{(x_0,y_0)}}$ for some $x_0$ (note that based on \eqref{eq:partial derivative of z wrt y}, $-\left.\frac{dz}{dy} \right\rvert_{{(x_0,y_0)}}$ is independent depend of $x_0$). Then, it follows from \eqref{eq:partial derivative of z wrt y} that $g(\cdot)$ is a strictly decreasing and smooth function, and
\begin{equation} \label{eq:g y smooth}
\begin{split}
g(y_0) \,\xrightarrow{\,n\to\infty\,}\,\begin{cases} -1 & y_0>0, \\1 & y_0<0.  \end{cases}
\end{split}
\end{equation}
Hence, for any $\epsilon>0$, there is a sufficiently large $\tilde{n}_\epsilon$, such that for any $n\ge \tilde{n}_\epsilon$ and any $x_0\in[-1/4,\,1/4]$, there exists a $y_0<\epsilon/3$, such that $g(y_0)=x_0$.
Let $y_0(\cdot)$ be the inverse of $g(\cdot)$, that is for any $x_0\in[-1/4,\,1/4]$,  $g\big(y_0(x_0)\big) = x_0$; equivalently,
\begin{equation}\label{eq:dzdy =x0}
-\left.\frac{dz}{dy} \right\rvert_{{(x_0,y_0(x_0))}} = x_0.
\end{equation}
It follows from the smoothness of $g(\cdot)$ that $y_0(\cdot)$ is smooth as well.

Let $n_\epsilon = \max\big(\tilde{n}_\epsilon,\, 1/\epsilon\big)$.
Then, for $\big(x_0,y_0(x_0)\big)$, plugging \eqref{eq:dzdx small}   and \eqref{eq:dzdy =x0} into \eqref{eq:scaled normal vec}, we obtain
\begin{equation} \label{eq:second def  for v in cartesian}
v = -\frac{2x_0}{1+2\pi n} \xcor + x_0\ycor + \zcor.
\end{equation}
Let $p = \big( x_0 ,\, 0,\, -2\pi n   \big)$. 
Then, 
\begin{equation} \label{eq:x0y0z in eps dist from p}
\begin{split}
\big\|\big( x_0 ,\, y_0(x_0),\,& z(x_0,y_0(x_0))\big) -p\big\|
\,\le\, \big|y_0(x_0)\big| \,+\, \big| z(x_0,y_0(x_0)) + 2\pi n\big|\\
&=\, \big|y_0(x_0)\big| \,+\, \frac1{2\pi n} \ln\Big(\cosh \big(2\pi n \,y_0(x_0)\big)\Big)   \,+\,\frac{x_0^2+y_0(x_0)^2}{1+2\pi n}\\
&\le\, \big|y_0(x_0)\big| \,+\, \big|y_0(x_0)\big|   \,+\,\frac{x_0^2+y_0(x_0)^2}{1+2\pi n}\\
&<\, \frac{\epsilon}{3} \,+\, \frac{\epsilon}{3} \,+\,  \frac{(1/4)^2+ (1/4)^2}{1+2\pi/\epsilon}\\
&<\, \epsilon,
\end{split}
\end{equation}
where the first inequality is a triangle inequality, the equality is from the definition of $z(x_0,y_0)$ in \eqref{eq:eq of z in x y}, 
the second inequality is because $\ln\big(\cosh(x)\big)\le |x|$, for all $x\in \R$, 
and the third inequality is due to the assumptions $y_0(x_0)<\epsilon/3$, $x_0,y_0(x_0)<1/4$, and $n\ge n_\epsilon\ge 1/\epsilon$.
Combining  \eqref{eq:grad x0y0z =av} and \eqref{eq:x0y0z in eps dist from p}, and letting $\alpha(p)=\alphah\big( x_0 ,\, y_0(x_0),\, z(x_0,y_0(x_0))\big)$ we obtain 
\begin{equation} \label{eq:vinfteps}
\alpha(p)\, v\in \tilde{F}_\epsilon(p).
\end{equation}
Since $\alphah(\cdot)$, $y_0(\cdot)$, and $z(\cdot)$ are smooth, so is $\alpha(\cdot)$.

Back to the cylindrical coordinates, letting $r_0=x_0$, we have $p = \big(r_0, -2\pi n, -2\pi n\big)$, and \eqref{eq:second def  for v in cartesian} turns into $v \,=\, -{2r_0}/{(1+2\pi n)} \rcor + r_0\phicor + \zcor$. Then, \eqref{eq:vinfteps} implies that
\begin{equation} \label{eq:v 3rd def}
-\frac{2r_0}{1+2\pi n} \rcor +r_0\phicor + \zcor\, \in\, \frac{1}{\alpha(p)} \tilde{F}_\epsilon(p).
\end{equation}
 
By rotation of the coordinates around the $z$ axis, we can make a similar argument for every $z\le -2\pi n_\epsilon$, which is not necessarily an integer multiple of $2\pi$. 
Then, letting $z_\epsilon = - 2\pi n_\epsilon$, the lemma follows from \eqref{eq:v 3rd def}.
\end{proof}


For any point $x_0$ on the $z$ axis, there is an unperturbed trajectory $x(\cdot)$ of $F$, which is also an unperturbed trajectory of $\tilde{F}_\epsilon$, initialized at $x(0)=x_0$, that always stays on the $z$-axis.
In what follows, we will use Lemma~\ref{claim:spread books} to show that there is a perturbed trajectory $\xt(\cdot)$ of $\tilde{F}_\epsilon$ corresponding to a perturbation of size $\epsilon$, that is initialized on the $z$ axis and whose distance from the $z$ axis grows larger than $1/6$ at some positive time.

Consider an auxiliary dynamical system $\dot{x}=G(x)$, with
\begin{equation} \label{eq:def aux dyn sys G}
 G\big(r,\phi,z\big)  = \frac{-2r}{1 - z}\rcor +  r\phicor + \zcor ,
\end{equation}
over the half-cylinder $\Omega$.
Let 
\begin{equation} \label{eq:def of z0 in lemma 2}
z_0 = z_\epsilon +\zeta - \frac{2}{\epsilon},
\end{equation}
where $z_\epsilon$ is the  constant in the statement of Lemma~\ref{claim:spread books}, and $\zeta$ is the constant in the statement of the theorem. 
Let $r(0)=0$, and let  $r:\R_+\to\R$ be the solution of the differential equation
\begin{equation} \label{eq:def of rt including epsilon}
\dot{r}(t) =\epsilon - \frac{2r}{1-z_0-t}.
\end{equation}
Then, for any $t\in[0,1/(2\epsilon)]$, we have $\dot{r}(t)<\epsilon$, and  as a result,
\begin{equation} \label{eq:rt le 1over2}
r(t)\,\le\, \epsilon t\,\le\, 1/2.
\end{equation}
Moreover, for any $t\in[0,1/(2\epsilon)]$,
\begin{equation} \label{eq:r dot le e/3}
\begin{split}
\dot{r}(t) \, &=\, \epsilon - \frac{2r}{1-z_0-t}\\
& \ge\, \epsilon - \frac{2\times (1/2)}{1-z_\epsilon -\zeta +2/\epsilon- 1/(2\epsilon)}\\
& >\, \epsilon - \frac{1}{3/(2\epsilon)}\\
& =\, \frac{\epsilon}{3},
\end{split}
\end{equation}
where the first inequality is due to \eqref{eq:rt le 1over2} and the definition of $z_0$ in \eqref{eq:def of z0 in lemma 2}.
Therefore, $r\big(1/(2\epsilon)\big) > 1/6$, and the second inequality is because $z_\epsilon,\zeta<0$ (as mentioned in the statements of Lemma~\ref{claim:spread books} and Theorem~\ref{th:3d F has unbounded sensitivity}).
It follows from \eqref{eq:r dot le e/3} that there exists a $t_0\in \big(0,1/(2\epsilon)\big)$ at which
\begin{equation}\label{eq:rt0 = 1 over 6}
\begin{split}
&r(t_0) \,=\, 1/6,\\
0\le \, & r(t)\,\le\, 1/6,\qquad \forall t\in [0,t_0].
\end{split}
\end{equation}
We fix this $t_0$ for the rest of the proof.

In cylindrical coordinates, for any  $t\in[0,t_0]$ let 
\begin{equation} \label{eq:def pt in cyl coordinates}
p(t) = \big(r(t), z_0+t, z_0+t \big).
\end{equation}
Also let $u(t)=\epsilon \rcor$ in the local cylindrical coordinates at $p(t)$, and
$\tilde{U}(t) = \int_{0}^t u(\tau)\,d\tau$.
Then, $u(t)=\epsilon  \cos(z_0+\tau) \xcor +  \epsilon  \sin(z_0+\tau) \ycor$ in the Cartesian coordinates, and
\begin{equation}\label{U tilde is bounded by 2 eps}
\begin{split}
\Ltwo{\tilde{U}(t)} \, &=\, \Ltwo{ \int_{0}^t u(\tau)\,d\tau}\\
&\le\, \Ltwo{ \int_{0}^t \epsilon \cos(z_0+\tau) \,d\tau} \,+\, \Ltwo{ \int_{0}^t \epsilon \sin(z_0+\tau)\,d\tau}\\
&\le\, 2\epsilon.
\end{split}
\end{equation}
Moreover, for any $t\in[0,t_0]$, in the local cylindrical coordinates at $p(t)$ we have
\begin{equation}\label{eq:dot p is u plus g}
\begin{split}
\dot{p}(t) \,&=\, \dot{r}(t) \rcor \,+\, r(t)\phicor \,+\, \zcor\\
&=\, \epsilon \rcor \,-\,  \frac{2r(t)}{1-z_0-t} \rcor \,+\, r(t)\phicor \,+\, \zcor\\
&=\, u(t) \,+\, G\big(p(t)\big).
\end{split}
\end{equation}
where the equalities are due to \eqref{eq:def pt in cyl coordinates}, \eqref{eq:def of rt including epsilon}, and \eqref{eq:def aux dyn sys G}, respectively. 
Then, $p(\cdot)$ is a perturbed trajectory of $G(\cdot)$ corresponding to perturbation function $\tilde{U}(\cdot)$.

Consider the function $\alpha(\cdot):\Omega \to \R_+$ defined in Lemma~\ref{claim:spread books}, and let $\beta(\cdot)$ be a solution of  the differential equation
\begin{equation} \label{eq:ode of beta}
\dot{\beta}(t) = \alpha\Big( p\big( \beta(t) \big) \Big).
\end{equation}
Since $\alpha(\cdot)$ and $p(\cdot)$ are smooth and locally Lipschitz functions,  $\alpha \circ p (\cdot)$ is also locally Lipschitz, and \eqref{eq:ode of beta} has a solution \cite{Butc16}. 

Let $t_1=\beta^{-1}(t_0)$, and let $\xt(t)=p\big(\beta(t)\big)$, for all $t\in[0,t_1]$.
Then,
\begin{equation}\label{eq:xt t1 p t0}
\xt(t_1)=p(t_0).
\end{equation}
For  $t\in[0,t_1]$, let $\xt_z(t)$ and $p_z\big(\beta(t)\big)$ be the $z$ coordinates of $\xt(t)$ and $p\big(\beta(t)\big)$, respectively.
Then, for any $t\in[0,t_1]$, 
\begin{equation} \label{eq:xtz less than zeps}
\begin{split}
\xt_z(t) \, &=\, p_z\big(\beta(t)\big)\\
&=\, z_0+\beta(t)\\ 
&\le\, z_0+\beta(t_1)\\ 
&=\, z_0+t_0\\
&=\, z_\epsilon +\zeta -\frac{2}{\epsilon} \,+\, t_0\\
&\le\, z_\epsilon +\zeta -\frac{2}{\epsilon} +\frac{1}{2\epsilon}\\
&\le\, z_\epsilon+\zeta .
\end{split}
\end{equation}
where 
the first inequality is because $\dot{\beta}>0$, 
fourth equality is by the definition of  $z_0$ in \eqref{eq:def of z0 in lemma 2}, 
and the second inequality is because $t_0\le 1/(2\epsilon)$.

For any $t\in [0,t_1]$, let $U(t) = \tilde{U}\big(\beta(t)\big)$. 
We now show that $\xt(\cdot)$ is a perturbed trajectory of $\tilde{F}_\epsilon$, corresponding to the perturbation function $U(\cdot)$. 
We have for any $t\in[0,t_1]$,
\begin{equation}\label{eq:xt dot as a sum of alpha g and U tilde}
\begin{split}
\frac{d}{dt}\xt(t) \,&=\, \frac{d}{dt} p\big( \beta(t) \big)\\
\,&=\, \dot{\beta}(t)\, \frac{d}{d\beta(t)} p\big( \beta(t) \big)\\
\,&=\, \dot{\beta}(t)\, G\Big( p\big( \beta(t) \big)\Big)  \, +\, \dot{\beta}(t)\, u\big(\beta(t)\big)\\
\,&=\, \alpha\Big(p\big(\beta(t)\big)\Big)\, G\big( \xt(t) \big)  \, +\, \frac{d}{dt}\tilde{U}\big(\beta(t)\big)\\
\,&=\, \alpha\big(\xt(t)\big)\, G\big( \xt(t) \big)  \, +\, \frac{d}{dt}U(t),
\end{split}
\end{equation}
where the last three equalities are due to \eqref{eq:dot p is u plus g}, \eqref{eq:ode of beta}, and the definition of $\xt(\cdot)$, respectively.
Moreover, it follows from Lemma~\ref{claim:spread books}, definition of $G(\cdot)$ in \eqref{eq:def aux dyn sys G}, and \eqref{eq:xtz less than zeps}, that $\alpha\big(\xt(t)\big)\, G\big( \xt(t) \big)\ \in \tilde{F}_\epsilon\big(\xt(t)\big)$, for all $t\in \big[0,t_1 \big]$.
Then, \eqref{eq:xt dot as a sum of alpha g and U tilde} implies that for any $t\in [0,t_1]$,
\begin{equation}
\frac{d}{dt}\xt(t) \,\in\,  \tilde{F}_\epsilon\big(\xt(t)\big)   \, +\, \frac{d}{dt}U(t).
\end{equation}
Therefore, $\xt(\cdot)$ is a perturbed trajectory of $\tilde{F}_\epsilon$, corresponding to the perturbation function $U(t)$. 

Finally, let $x(\cdot)$ be an unperturbed trajectory of $\tilde{F}$, initialized at $x(0) = \xt(0)=\big(0,z_0,z_0\big)$, that always stays on the $z$ axis.
Then, \eqref{eq:xt t1 p t0} and \eqref{eq:rt0 = 1 over 6} imply that 
\begin{equation} \label{eq:xt -x t =1/6}
\Ltwo{\xt(t_1) -x(t_1)} \,=\, \Ltwo{p(t_0) -x(t_1)}
\,\ge\, r(t_0)
\, = {1}/{6}.
\end{equation}
Moreover, it follows from \eqref{U tilde is bounded by 2 eps} that 
\begin{equation}\label{eq:ltwo U le 2eps}
\sup_{\tau\le t_1}\Ltwo{U(t)} \,=\, \sup_{\tau\le t_1}\Ltwo{\tilde{U}\big(\beta(t)\big)} \,\le\, 2\epsilon.
\end{equation} 
Since for any $\epsilon>0$, there exists such pair $x(\cdot)$ and $\ptraj(\cdot)$ of trajectories and perturbation $U(\cdot)$ for which \eqref{eq:xt -x t =1/6} and \eqref{eq:ltwo U le 2eps} hold, no constant $C$ can satisfy \eqref{eq:bis of spread systems}.
Then, Theorem~\ref{th:spread sys} implies that  $F$ has unbounded sensitivity. This completes the proof of Theorem~\ref{th:3d F has unbounded sensitivity}.

\medskip


\medskip 
\section{\bf Proof of Theorem~\ref{th:pwc example}} \label{app:pwc example}
To prove Theorem~\ref{th:pwc example}, we follow the high level idea discussed in Section~\ref{subsec:pwc with infinite pieces}.
Consider the convex potential function $\Phi$ of Section~\ref{subsec:subdiff dont work} over the half-cylinder $\Omega$ (defined in \eqref{eq:def Omega}), whose gradient field $F$ has unbounded sensitivity.
It was shown in Theorem~\ref{th:3d F has unbounded sensitivity} that for any $\zeta\le -1$, $F$ has unbounded sensitivity over $\Omega_\zeta$ (defined in \eqref{eq:def Omega zeta}).
We now construct a sequence $\zeta_i$, $i=1,2,\ldots$, of negative numbers as follows.
Let $\zeta_0 =-1$. 
For any $i\in\Z_+$, Theorem~\ref{th:3d F has unbounded sensitivity} implies that there exist an $\epsilon_i>0$, a time $t_i>0$, 
an unperturbed trajectory $x_i(\cdot)$ of $F$, and a perturbed trajectory $\xt_i(\cdot)$ of $F$ corresponding to a perturbation function $U_i(\cdot)$ 
and initialized at $\xt_i(0)=x_i(0)$, such that 
\begin{equation}\label{eq:x xt in omega}
x_i(t), \xt_i(t)\in \Omega_{\zeta_i},\quad \forall t\in[0,t_i],
\end{equation}
\begin{equation}
\big\| U_i(t) \big\|  \le \epsilon_i,  \quad \forall t\in[0,t_i],
\end{equation}
and
\begin{equation}
\big\| \xt_i(t_i)-x_i(t_i) \big\|  \ge  i \, \epsilon_i.
\end{equation}
We then let for any $i\in\Z_+$,
\begin{equation}
\zeta_{i+1} = -\left( 1+ \sup_{t\in[0,t_i]}  \big\| \xt_i(t) \big\| \, + \, \sup_{t\in[0,t_i]}  \big\| x_i(t) \big\| \right) .
\end{equation}
Then, it follows from \eqref{eq:x xt in omega} that $\xt_i(0)\le\zeta_i$, and therefore $\zeta_{i+1}<\zeta_i$. Hence, $\zeta_i$, $ i=0,1,\ldots$, is a decreasing sequence. 
Therefore, $\Omega_{\zeta_{i+1}}\subset \Omega_{\zeta_i}$, for all $i\in\Z_+$.
For any $i\in\Z_+$, let 
\begin{equation}
\Gamma_i  \, = \, \Omega_{\zeta_i} \backslash \Omega_{\zeta_{i+1}} \, = \, \Big\{(r,\phi,z) \,\Big|\, r\le 1/6, \, z  \in \big[\zeta_{i+1}, \zeta_i\big) \Big\}.
\end{equation}
Then, each $\Gamma_i$ is bounded. Moreover, for any $t\in[0,t_i]$, we have $x_i(t), \xt_i(t)\in \Gamma_i$.

We now construct a fine grid within the half-cylinder $\Omega$, defined in \eqref{eq:def Omega}, with increasing resolution as $z\to-\infty$, and consider a corresponding triangulation of $\Omega$ with simplexes whose vertices lie at these grid points. 
We then consider a piecewise linear approximation $\Psi$ of $\Phi$, by letting $\Psi(x)=\Phi(x)$ on the grid points, and let  $\Psi$ be a linear interpolation inside each simplex. 
Since $\Phi$ is continuously twice differentiable and  $\Gamma_i$ is compact, for all $i\in\Z_+$, we can  choose the grid points so that $\Psi$ gives an arbitrarily accurate approximation of $\Phi$ inside $\Gamma_i$, in the following sense: 
for any $i\in\Z_+$, any point $p\in\Gamma_i$, and any $\xi\in\partial \Psi(p)$,
\begin{equation} \label{eq:approx of Phi with Psi is good and grads are close}
\big\|  \xi - \grad\Phi(p)  \big\|  \,\le\, \frac{\epsilon_i}{t_i}.
\end{equation}
Capitalizing on \eqref{eq:approx of Phi with Psi is good and grads are close}, we now show that $\xt_i(\cdot)$ and $x_i(\cdot)$ are perturbed trajectories of the gradient field of $\Psi$, corresponding to perturbations of size no larger than $2\epsilon_i$.

Let $H$ be the (sub)gradient field of $\Psi$, and fix an $i\in\Z_+$. 
For any $t\in [0,t_i]$,
\begin{equation}\label{eq:x i decomposition on xi and grad phi}
\begin{split}
x_i(t) \,&=\, x_i(0) \, +\, \int_{0}^t \grad\Phi\big(x_i(\tau)\big) \,d\tau \\
&=\, x_i(0) \, +\, \int_{0}^t \xi\big(x_i(\tau)\big) \,d\tau \,+\,  \int_{0}^t \Big(\grad\Phi\big(\xt(\tau)\big) - \xi\big(x_i(\tau)\big)\Big) \,d\tau,
\end{split}
\end{equation}
where $\xi\big(x_i(\tau)\big)$ is an arbitrary subgradient of $\Psi $ at $x_i(\tau)$.
In the same vein, for any $t\in [0,t_i]$,
\begin{equation}\label{eq:xt i decomposition on xi and grad phi}
\begin{split}
\xt_i(t) \,&=\, \xt_i(0) \, +\, \int_{0}^t \grad\Phi\big(\xt_i(\tau)\big) \,d\tau \,+\, U(t)\\
&=\, \xt_i(0) \, +\, \int_{0}^t \tilde{\xi}\big(\xt_i(\tau)\big) \,d\tau \,+\, \bigg[ U(t) \, +\, \int_{0}^t \Big(\grad\Phi\big(\xt_i(\tau)\big) - \tilde{\xi}\big(\xt_i(\tau)\big)\Big) \,d\tau\bigg],
\end{split}
\end{equation}
where $\tilde{\xi}\big(\xt_i(\tau)\big)$ is an arbitrary subgradient of $\Psi $ at $\xt_i(\tau)$.
For $t\in[0,t_i]$, let
\begin{equation}
V_i(t) = \int_{0}^t \Big(\grad\Phi\big(x_i(\tau)\big) - \xi\big(x_i(\tau)\big)\Big) \,d\tau,
\end{equation}
and 
\begin{equation}
\tilde{U}_i(t) =  U_i(t) \, +\, \int_{0}^t \Big(\grad\Phi\big(\xt_i(\tau)\big) - \tilde{\xi}\big(\xt_i(\tau)\big)\Big) \,d\tau.
\end{equation}
It follows from \eqref{eq:x i decomposition on xi and grad phi} and \eqref{eq:xt i decomposition on xi and grad phi} that $x_i(\cdot)$ and $\xt_i(\cdot)$ are perturbed trajectories of $H$ corresponding to perturbations $V_i(\cdot)$ and $\tilde{U}_i(\cdot)$, respectively.
Moreover, for any $t\in[0,t_i]$,
\begin{equation*}
\begin{split}
\big\|  V_i(t) \big\| &\,=\, \Ltwo{ \int_{0}^t \Big(\grad\Phi\big(x_i(\tau)\big) - \xi\big(x_i(\tau)\big)\Big) \,d\tau}   \\
&\,\le\, \int_{0}^t \Ltwo{\grad\Phi\big(x_i(\tau)\big) - \xi\big(x_i(\tau)\big)} \,d\tau\\
&\,\le\,  \int_{0}^t \frac{\epsilon_i}{t_i} \,d\tau\\
&\,=\, \frac{t\,\epsilon_i}{t_i}\\
&\, < \, 2\epsilon_i,
\end{split}
\end{equation*}
where the second inequality is due to  \eqref{eq:approx of Phi with Psi is good and grads are close}. In the same vein,
\begin{equation*}
\begin{split}
\big\|  \tilde{U}_i(t) \big\| &\,=\, \Ltwo{ U_i(t) \,+\, \int_{0}^t \Big(\grad\Phi\big(\xt_i(\tau)\big) - \xi\big(\xt_i(\tau)\big)\Big) \,d\tau}   \\
&\,\le\, \Ltwo{U_i(t)} \,+\,\int_{0}^t \Ltwo{\grad\Phi\big(\xt_i(\tau)\big) - \xi\big(\xt_i(\tau)\big)} \,d\tau\\
&\,\le\,  \epsilon_i \,+ \,\int_{0}^t \frac{\epsilon_i}{t_i} \,d\tau\\
&\,=\, \epsilon_i \,+ \,\frac{t\,\epsilon_i}{t_i}\\
&\, \le \, 2\epsilon_i.
\end{split}
\end{equation*}

Let $y_i(\cdot)$ be an unperturbed trajectory of $H$ initialized at $y_i(0)=x_i(0)$.
\oli{The fact that $y_i(t)$ does not leave $\Omega$ before time $t_i$ can be proved, given that the grid is fine enough. We can choose the grid so fine for this to happen.}
Then,
\begin{equation*}
\begin{split}
\Ltwo{ \xt_i(t_i) -y_i(t_i)} \, +\, \Ltwo{ x_i(t_i) -y_i(t_i)} &\, \ge\, \Ltwo{ \xt_i(t_i) -x_i(t_i)}
\, \ge\, i\,\epsilon_i.
\end{split}
\end{equation*}
Therefore, either $\Ltwo{ \xt_i(t_i) -y_i(t_i)}\ge i\epsilon_i/2$ or  $\Ltwo{ x_i(t_i) -y_i(t_i)}\ge i\epsilon_i/2$.
Hence, for any $i\in\Z_+$, there is a perturbation function of size no larger than $2\epsilon_i$, and a corresponding pair of perturbed and unperturbed trajectories of $H$, with the same initial conditions, whose distance grows larger than $i\epsilon_i/2$ in time $t_i$.
This implies that $H$ has unbounded sensitivity and completes the proof of the theorem.


\medskip
\section{\bf  Proof of Theorem~\ref{th:main dis}} \label{sec:proof disc}
The high level idea is to simulate a discrete time perturbed trajectory with a continuous time perturbed trajectory, 
and then take advantage of the bounded sensitivity property of the  continuous time system.


\begin{lemma}[Simulation of Discrete Time Trajectories with Continuous Time Trajectories] \label{lem:simul}
Consider a dynamical system  $F$, a discrete time perturbation $\dpert(\cdot)$, a corresponding discrete time trajectory  $\dtraj(\cdot)$, and a $\ddrift(\cdot)$ that satisfies \eqref{eq:def of discrete traj}.
Let $\pert(\cdot)$ be a continuous time perturbation,
\begin{equation} \label{eq:simul bound pert 1}
\pert(t) \,= \, \dpert\big(\tf\big) \, - \, \big(t-\tf\big) \ddrift\big(\tf\big),\qquad \forall t\in\R^n.
\end{equation}
Then, there exists a  corresponding continuous time perturbed trajectory $\ptraj(\cdot)$  such that
\begin{equation}\label{eq:simul x z equal 1}
\ptraj(k)=\dtraj(k), \qquad  k=0,1,\ldots\,.
\end{equation}
\end{lemma}
\begin{proof}
For any $t\in\R_+$, let
\begin{equation}
\ptraj(t)=\dtraj\big(\tf\big)
\end{equation}
We show that $\ptraj$ is a perturbed trajectory corresponding to perturbation $\pert$.
For any $t\in\R_+$, let $\drift(t)= \ddrift\big(\tf\big)$. 
Then, $\drift(t)\in F\big( \dtraj(\tf)  \big) = F\big( \ptraj(t)  \big)$. 
Moreover, for any $t\in\R_+$,
\begin{equation}
\begin{split}
\ptraj(t) \, &=\, \dtraj\big(\tf\big)\\
&=\, \sum_{k\le t} \ddrift(k) \, +\, \dpert\big(\tf\big)\\
&=\, \sum_{k\le t} \int_{k-1}^k \drift(\tau)\,d\tau \, +\, \Big[ \pert(t) + \big(t-\tf\big) \ddrift\big(\tf\big)\Big]\\
&=\, \int_{0}^{\tf} \drift(\tau)\,d\tau \, +\, \pert(t) + \int_t^{\tf} \drift(t)\\
&=\, \int_{0}^{t} \drift(\tau)\,d\tau \, +\, \pert(t).
\end{split}
\end{equation}
Therefore,  $\ptraj$ is a perturbed trajectory corresponding to perturbation $\pert$, and the lemma follows.
\end{proof}

\medskip
\begin{proof}[\bf Proof of Theorem \ref{th:main dis}]
For Part (a), 
 it follows from Lemma \ref{lem:simul} that there exists a perturbation function $\pert(\cdot)$ with corresponding continuous time perturbed trajectory $\ptraj(\cdot)$ that satisfy  (\ref{eq:simul bound pert 1}) and (\ref{eq:simul x z equal 1}). 
 Then, for any $k\in\Z_+$,
\begin{equation}\label{eq:proof of th disc first case bound}
\begin{split}
\Ltwo{x(k)-\dtraj(k)} \,&= \, \Ltwo{x(k)-\ptraj(k)}\\
&\le \, \divconst\, \sup_{t\le k} \Ltwo{ \pert(\tau)}\\
&\le \,  \divconst\, \left(\max_{j< k} \Ltwo{ \ddrift(j)} \,+ \,\max_{j< k} \Ltwo{\dpert(j)}\right),
\end{split}
\end{equation}
where the relations are due to (\ref{eq:simul x z equal 1}),  \eqref{eq:bis}, and (\ref{eq:simul bound pert 1}), respectively. 
This completes the proof of Part~(a).

For Part~(b), consider, from Lemma~\ref{lem:simul}, a pair $\pert_1(\cdot)$ and $\pert_2(\cdot)$ of perturbations and a corresponding pair $\ptraj_1(\cdot)$ and $\ptraj_2(\cdot)$  of perturbed trajectories
such that for $i=1,2$,
\begin{equation}\label{eq:simul x z equal 2}
\ptraj_i(k)=\dtraj_i(k), \qquad \forall k\in\Z_+,
\end{equation}
\begin{equation} \label{eq:simul bound pert 2}
\pert_i(t) \,= \, \dpert_i\big(\tf\big) \, - \, \big(t-\tf\big) \ddrift_i\big(\tf\big),\qquad \forall t\in\R^n.
\end{equation}
Then, for any $k\in\Z_+$,
\begin{equation}\label{eq:proof of th disc first case bound 2}
\begin{split}
\Ltwo{\dtraj_1(k)-\dtraj_2(k)} \,&= \, \Ltwo{\ptraj_1(k)-\ptraj_2(k)}\\
&\le \, \divconst\, \sup_{t\le k} \Ltwo{ \pert_1(t) - \pert_2(t)}\\
&\le \,  \divconst\, \left(\max_{j< k} \Ltwo{ \ddrift_1(j) - \ddrift_2(j)} \,+ \,\max_{j< k} \Ltwo{\dpert_1(j)-\dpert_2(j)}\right),
\end{split}
\end{equation}
where the relations are due to (\ref{eq:simul x z equal 2}), \eqref{eq:bpis}, and (\ref{eq:simul bound pert 2}), respectively. 
This completes the proof of the theorem.
\end{proof}

\medskip
\section{\bf Discussion}\label{sec:discussion}
We studied boundedness of sensitivity to cumulative perturbations for some classes of dynamical systems of interest,
a property that, when holds, provides strong tools to analyze systems driven by stochastic noise.
We derived a necessary and sufficient condition for bounded sensitivity  of a linear dynamical system, in terms of its spectrum.  
More specifically, we showed that a linear system has bounded sensitivity  if and only if it is stable and has no periodic orbits.

Another class we studied was the class of dynamical systems driven by the gradients of a convex potential function. We showed that there exist subgradient fields of strictly convex as well as piecewise linear convex functions that have unbounded sensitivity.
This result is particularly important because it certifies the necessity of ``finiteness'' assumption in a former result (Theorem~\ref{th:p1 sensitivity}), 
according to which the subgradient field of a piecewise linear convex function with finitely many pieces has bounded sensitivity.

We also studied transformations of a dynamical system that preserve the bounded sensitivity property.
In particular, we showed for a dynamical system with bounded sensitivity that a similar property holds for its induced discrete time systems, spread systems, and the systems obtained via convolution with a kernel.

In the rest of this section, we discuss open problems and directions for future research.	

\begin{itemize}[label=$\bullet$]
\item 
Study of  sensitivity bounds to cumulative perturbations for other classes of dynamical systems: 
Since the sensitivity  bounds of types \eqref{eq:bis} and \eqref{eq:bpis} involve pretty strong inequalities, we do not expect that every class of systems can meet these bounds. 
However, when a system meets these bounds, it enables to make fine-grained analyses its perturbed dynamics. 
Therefore, it would be useful to investigate the bounded sensitivity property for classes of systems of practical importance.
For example, it is interesting to establish whether or not the class of nonexpansive piecewise linear dynamical systems (with all pieces being SOF) has bounded sensitivity.
This class generalizes the class of FPCS systems as well as the class of SOF linear systems.

\item 
Combination of two systems:
Besides fining more transformations of a system that preserve bounded sensitivity, an interesting research direction concerns bounded sensitivity of combination of two or several systems.
For example, for two systems with bounded sensitivity, does their pointwise summation  still have bounded sensitivity?
The answer is already known to be negative \oli{I have counterexamples}. 
However, the pointwise sum of systems with specific structures might yield bounded sensitivity. A prominent example is the sum of an FPCS system and a linear system. 
These combinations include the systems that underlie the gradient flows of the LASSO cost function, i.e., $\ltwo{x}_1 +\Ltwo{Ax-b}_2$.

Another direction involves preservation of bounded sensitivity under other means of combination of two dynamical systems, more general than a pointwise sum.
An interesting combination involves dynamical systems  over disjoint domains, that are glued together in a nonexpansive manner.

\item
Strong vs weak bounded sensitivity: 
As mentioned before, a sensitivity bound of type \eqref{eq:bpis} implies the sensitivity bound of type \eqref{eq:bis}.
The reverse however is not known.
It remains open to obtain the conditions  under which a dynamical system with a bound of type \eqref{eq:bis} has bounded sensitivity also in the sense of \eqref{eq:bpis}.
As a concrete example, 
Theorem~\ref{th:p1 sensitivity} shows that FPCS systems have bounded sensitivity in the sense of  \eqref{eq:bis}. 
It would be interesting if one could prove or disprove analogues results in the strong sensitivity sense of \eqref{eq:bpis}.

\item
Bounded domain:
We considered dynamical systems defined over the entire $\R^n$.
This however is not the case in many applications such as queueing  networks where the state space (i.e., queue lengths) is restricted to the positive orthant. 
Boundedness of the domain gives rise to boundary conditions like projecting the ``escaping trajectories''  back onto the domain, which typically further complicate the dynamics. 
In particular, one can investigate what types of boundary conditions will preserve bounded sensitivity, once the domain is restricted.

\hide{
This is however not the case in many applications such as communication networks where the state space (i.e., queue lengths) is a subset of $\R^n$ with non-negative coordinates.  
In this case, existence of boundary conditions further complicates the problem. 
Let $\Omega$ be the intersection of a finite number of half-spaces. 
For any point $x\in \Omega$ and any vector $v\in\R^n$, denote by $\pi^{\partial\Omega(x)}(v)$ the projection of $v$ on the conical hull of $\Omega-x$. 
Define the projected unperturbed and perturbed trajectories as:
\begin{equation}
\dot{x}(t) \in \pi^{\partial\Omega(x)}\big(F(x)\big),
\end{equation}
and \comment{see 'projected dynamical systems' in the literature}
\begin{equation}
\frac{d}{dt}\ptraj(t) \in \pi^{\partial\Omega(\ptraj)}\Big(F\big(\ptraj(t)\big) + \pert(t)\Big),
\end{equation}
respectively. We can ask the same question of  sensitivity for these restricted dynamical systems, e.g., whether or not Theorem \ref{th:main cont}  holds for FPC subgradient dynamical systems  restricted to $\Omega$. 
}

\item
Convolution by a kernel:
We showed in Corollary~\ref{th:kernel} that a system with bounded sensitivity, when convolved with a kernel, still satisfies a weaker notion of bounded sensitivity that incorporates additive penalties. 
However, it remains open that under what conditions on the kernel, the latter system would have bounded sensitivity in the sense of \eqref{eq:bis}, with no additive penalties.

\oli{ One more problem (I think this one should be easy.)
\\
Subgradient systems:
It was shown in \cite{AlTG19sensitivity} that if a subgradient system is piecewise constant with finitely many pieces, then the system has bounded sensitivity.
It would be interesting if one could establish the bounded sensitivity of subgradient systems under other practical conditions.
For example, the systems in Examples~\ref{ex:subdiff dont work} and \ref{ex:pwc with infinite pieces} both suffer from unbounded gradients as $z$ goes to infinity.
However, assuming a uniform bound on the gradients, it remains open to prove/disprove  bounded sensitivity of subgradient fields.
}

\end{itemize}


\medskip
\section*{Acknowledgment}
The authors thank John N. Tsitsiklis for fruitful discussions and insightful comments throughout  the course of development of this work. 

\section*{Conflict of Interest}
The authors declare that they have no conflict of interest.

\bibliography{schbib}
\bibliographystyle{ieeetr}

\newpage

\appendix

\section{\bf Proof of Lemma \ref{lem:properties of f}}  \label{app:proof wellformedness of f}
For Part (a), consider function $h:\R^4 \to \R$, with 
\begin{equation} 
h\big(f,r,\phi,z\big) = f\,+\, z \,-\, \frac1f\, \ln\Big(\cosh \big(fr\sin(f-\phi)\big)\Big)\, -\,  \frac{r^2}{1+f}, 
\end{equation}
for all $f>0$, and all $r,\phi,z\in\R$.
Then, $f\big(r,\phi,z\big)$ in \eqref{eq:book f (smooth)} is a solution of the implicit equation $h\big(f,r,\phi,z\big)=0$.
For any $\big(r,\phi,z\big)\in\Omega$,
\begin{equation} \label{eq:h1<0}
\begin{split}
h\big(1,r,\phi,z\big) \,&=\, 1\,+\, z  \,-\, \ln\Big(\cosh \big(r\sin(1-\phi)\big)\Big) \,-\, r^2  \,\le\, 1\,+\, z  \,\le\, 0
\end{split}
\end{equation} 
and 
\begin{equation}\label{eq:h inf=inf} 
\begin{split}
\lim_{f\to\infty} h\big(f,r,\phi,z\big) \,&=\, z \,+\, \lim_{f\to\infty}\left(f\,-\, \frac1f\, \ln\Big(\cosh \big(fr\sin(f-\phi)\big)\Big) \right) \\
& \ge \, z \,+\, \lim_{f\to\infty}\left(f\,-\, \frac{fr\sin(f-\phi)}f \right) \\
&=\, \infty,
\end{split}
\end{equation}
where the last inequality is because $\ln\big(\cosh(x)\big) \le |x|$, for all $x\in \R$.
Then, for any  $\big(r,\phi,z\big)\in\Omega$, there is an $f\ge1$ for which $h\big(f,r,\phi,z\big)=0$. 
We now prove the uniqueness of this $f$, by showing that for any  fixed  $\big(r,\phi,z\big)\in\Omega$, $h\big(f,r,\phi,z\big)$ is a strictly increasing function in $f$.

We have
\begin{equation} \label{eq:derivative of h wrt f is posetive}
\begin{split}
\frac{\partial}{\partial f}\,h\big(f,r,\phi,z\big) \, &=\, 1\,+\, \frac1{f^2}\, \ln\Big(\cosh \big(fr\sin(f-\phi)\big)\Big)\\
&\qquad -\,  \frac{r\sin(f-\phi) + fr\cos(f-\phi)}{f}\,\tanh\big(fr\sin(f-\phi)\big)     \,+\, \frac{r^2}{\big(1+f\big)^2}  \\
&\ge\, 1\,-\, \frac{r + fr}{f}\,\tanh\big(fr\big) \\
&=\, 1\,-\, \frac{r \,\tanh\big(fr\big)}f  \,-\,r\, \tanh\big(fr\big) \\
&\ge\, 1-r^2-r \\
&>\,0.
\end{split}
\end{equation}
where the the first inequality is by removing the positive terms and the trigonometric functions, the second inequality is  because $\tanh(x) \le x$ and $\tanh(x) \le 1$, for $x>0$, and the last inequality is from $\big(r,\phi,z\big)\in\Omega$.
Then, $h$ is a strictly increasing function in its first argument. Together with \eqref{eq:h1<0} and \eqref{eq:h inf=inf}  this implies that  for any  fixed  $\big(r,\phi,z\big)\in\Omega$, there is a unique $f$ that satisfies  $h\big(f,r,\phi,z\big)=0$.
This completes the proof of Part~(a).

\hide{
\begin{clm}\label{claim:inequalities of 1/4}
\begin{itemize}
\item[a)] For any $x\ge 0$,
\begin{equation}
\frac{\big(1/4\big)^{1/(1+x)}}{(1+x)^2} \,\le\, \frac{1}{2}
\end{equation}

\item[b)] For any $x\in [0,1]$,
\begin{equation}
(1+x)\, \big(1/4\big)^{1+x} \,\le\, \frac{3}{8}
\end{equation}
\end{itemize}
\end{clm}
\begin{proof}[Proof of Claim]
For Part (a), for $x\ge 1$,
$$ \frac{\big(1/4\big)^{1/(1+x)}}{(1+x)^2}  \,\le\,   \frac{\big(1/4\big)^{1/(1+\infty)}}{(1+1)^2}    \,=\, \frac{1}{4},  $$
and for $x\in [0,1]$,
$$ \frac{\big(1/4\big)^{1/(1+x)}}{(1+x)^2}  \,\le\,   \frac{\big(1/4\big)^{1/(1+1)}}{(1+0)^2}    \,=\, \frac{1}{2} .  $$
For Part (b), for $x\in \big[1/2 ,1\big]$,
$$ (1+x)\, \big(1/4\big)^{1+x} \,\le\,   (1+1)\, \big(1/4\big)^{1+1/2}   \,=\,   \frac{1}{4},  $$
and  for $x\in \big[ 0, 1/2 \big]$,
$$ (1+x)\, \big(1/4\big)^{1+x} \,\le\,   \big(1+1/2\big)\, \big(1/4\big)^{1+0}   \,=\,   \frac{3}{8}.  $$
\end{proof}

We now compute the derivative of $h$ and show that it is positive:
\begin{equation*}
\begin{split}
\frac{dh}{df}  \, &= \, 1 \,- \, \Bigg[ \mathrm{sgn}\big(\sin(f-\phi)\big)\, \left(1+ \frac{1}{1+f} \right) \, r\,\cos\big(f-\phi\big)\, \big\lvert r\sin\big(f-\phi\big) \big\rvert^{1/(1+f)}  \\
& \qquad\qquad - \,\frac{\big\lvert r\sin\big(f-\phi\big) \big\rvert^{1+1/(1+f)} \,\log \big\lvert r\sin\big(f-\phi\big) \big\rvert}{\big(1+f\big)^2} 
 \Bigg]
\, + \frac{r^2}{(1+f)^2}\\
& \ge\, 1 \,-\, \left(1+ \frac{1}{1+f} \right) \, r^{1+1/(1+f)}  \,-\, \frac{r^{1/(1+f)}}{\big( 1+f \big)^2} \, + 0  \\
& \ge\, 1 \,-\, \frac12 \, -\frac{3}{8} \\
& >0.
\end{split}
\end{equation*}
where the first inequality is because  $\big\lvert r\sin\big(f-\phi\big)\,\log \big\lvert r\sin\big(f-\phi\big)\ge -1$,  and the second  inequality is due to Claim~\ref{claim:inequalities of 1/4}.
Then, $h$ is an strictly increasing function of $f$ and its zero-crossing is unique.
This completes the proof of Part~(a).
}

Part (b) is immediate from the definition of $f$ in  \eqref{eq:book f (smooth)}. 
For Part (c), 
note that $\frac{d^2}{dx^2}\ln\cos (x) = 1/\cosh(x)^2\ge 0$, and $\ln\cosh (\cdot) $ is thereby a convex function.
Then, the surface of each level-set of $f$, given in \eqref{eq:surf eq}, is the sum of a convex function, $z\big(r,\phi\big) = -f\,+\, \, \ln\Big(\cosh \big(fr\sin(f-\phi)\big)\Big)\,/f$, and a strictly convex function, $z\big(r,\phi\big) =r^2/(1+f)$. 
Hence, the surface, $z$, of each level-set is a strictly convex function. 
Then, for any pair of points $p_1,p_2\in \Omega$  with $f(p_1)=f(p_2)=a$, the line segment connecting $p_1$ and $p_2$ lies above the level set $f = a$. 
Equivalently, for any $\alpha\in(0,1)$, $f\big(\alpha p_1+(1-\alpha)p_2\big) > a $.
Thus, $f$ is strictly quasi-convex.
Moreover,  each level-set of $f$ is the intersection of $\Omega$ with  a  surface of type \eqref{eq:surf eq}, and is thereby  compact. 
For smoothness, note that $f$ is a solution of the implicit equation $h\big(f,r,\phi,z\big)=0$, where $h$ is smooth and, from \eqref{eq:derivative of h wrt f is posetive}, $\partial h/\partial f >0$.
Then, it follows from the ``implicit function theorem'' for smooth functions \cite{Shahshah16} (Theorem 12 of Appendix B) that $f(\cdot)$ is smooth.

\medskip


\section{\bf Proof of Lemma~\ref{lem:simul spread}}\label{app:proof spread}
By the definition of a perturbed trajectory,
\begin{equation}\label{eq:the drift and int eq}
\begin{split}
\ptraj(t) \,&=\, \int_{0}^t \drift(\tau)\,d\tau \, + \, \pert(t), \quad \forall t\ge0,\\
\drift(t) \,&\in\, \tilde{F}_\epsilon\big(\ptraj(t)\big), \quad \forall t\ge0.
\end{split}
\end{equation}
The term $\int_{0}^t \drift(\tau)\,d\tau$ in the equality, is a continuous function of $t$ and $\pert(\cdot)$ is a right-continuous function. 
Then, $\ptraj(\cdot)$ is right continuous.

In the proof that follows, we use a \emph{transfinite recursion} \cite{HrbaJ99} to partition $\R$ into a number of time intervals $[t_i, t_{i+1})$.
Let $\ord$ be the collection of all \emph{ordinal numbers}  \cite{HrbaJ99}.
Consider the sequence $t_i$ defined by the following transfinite recursion:
\begin{itemize}[label=$\bullet$]
\item Base case: 
$t_0=0$.

\item Successor case:
For any successor ordinal $\alpha$, let 
\begin{equation}\label{eq:tf t suc}
\begin{split}
t_\alpha\, = \, &\min\bigg( t_{\alpha-1}+ \frac{\delta}{\lip\,\big(\ltwo{\ptraj(t_{\alpha-1})} +\epsilon+\delta\big)+1}\, ,\\ 
&\qquad\qquad \sup\left\{t\in\R\,\Big|\, \ptraj(\tau) \in \ball_\delta\big(\ptraj({t_{\alpha-1}})\big)  ,\,\forall \tau\in[t_{\alpha-1},t) \right\}    \bigg),
\end{split}
\end{equation}
where $\lip$ is the constant in \eqref{eq:constant not to blow up}. 

\item Limit case: For any limit ordinal $\alpha$, let
\begin{equation}\label{eq:tf t lim}
t_\alpha\, = \,\sup\big\{ t_\beta \, \big|\, \beta \in\ord,\, \beta<\alpha \big\}.
\end{equation}

\item Termination: If $t_\alpha=\infty$, halt and let $\alphamax=\alpha$.
\end{itemize}

\begin{clm} \label{claim:alpha max exists}
The ordinal $\alphamax$ exists and is a limit ordinal. 
Moreover, the intervals $[t_\alpha, t_{\alpha+1})$, $\alpha<\alphamax$, cover  $\R_+$.
\end{clm}

\oli{why did we use transfinite recursion instead of a simple recursion/induction on the natural numbers?
Answer: In that case, $t_i$'s might have ended up converging to some $T<\infty$, and the intervals  not covering $\R$.}

\begin{proof}[Proof of  Claim \ref{claim:alpha max exists}]
It follows  from \eqref{eq:tf t suc} and the right-continuity of $\ptraj(\cdot)$ that for any successor ordinal $\alpha$, $t_\alpha>t_{\alpha-1}$.
Together with \eqref{eq:tf t lim}, this implies that for any  ordinal $\alpha<\alphamax$, and any ordinal $\beta<\alpha$, $t_\beta<t_{\alpha}$.
Then, all values of $t_\alpha$ are distinct, and the length of the sequence $t_\alpha$, $\alpha\le \alphamax$, can be no larger than the cardinality of $\R$, i.e., $2^{\aleph_0}$.

Assuming the ``axiom of choice'', the ``Von-Neumann's cardinal assignment'' \cite{Mosc06} implies that $2^{2^{\aleph_0}}$ equals some ordinal number $\beta$.
Then, since $2^{\aleph_0} < 2^{2^{\aleph_0}}= \beta$, the transfinite recursion defining $t_\alpha$ must terminate for some value of $\alpha<\beta$.
Hence, $\alphamax$ exists and is less than $\beta$.
Moreover, $\alphamax$ cannot be a successor ordinal, because in that case, $t_{\alphamax-1}<\infty$ and \eqref{eq:tf t suc} would have implied that $t_\alphamax<\infty$.

For the second part of the claim, note that
\begin{equation}
\bigcup_{\alpha<\alphamax} \big[t_{\alpha},t_{\alpha+1}\big) 
\, =\, \sup_{\alpha<\alphamax} \big[t_{0},t_{\alpha}\big) 
\, =\,  \big[t_{0},t_{\alphamax}\big) 
\, =\,  [0,\infty) 
\, =\,  \R_+,
\end{equation}
and the claim follows.
\end{proof}

We continue by defining the perturbation $\pert'(\cdot)$ and its corresponding perturbed  trajectory $\yt(\cdot)$.
Fix some $\alpha<\alphamax$.
It follows from \eqref{eq:tf t suc} that for any $t\in [t_\alpha,t_{\alpha+1})$, $\Ltwo{\ptraj(t)-\ptraj(t_\alpha)}\le\delta$.
Then, \eqref{eq:the drift and int eq} implies that for any $t\in [t_\alpha,t_{\alpha+1})$,
\begin{equation}
\drift(t)\, \in\, \tilde{F}_\epsilon\big(\ptraj(t)\big)
\, \subseteq\, \tilde{F}_{\epsilon+\delta}\big(\ptraj(t_\alpha)\big).
\end{equation}
Therefore, 
\begin{equation}
\begin{split}
\frac{1}{t_{\alpha+1}-t_\alpha} \,\int_{t_\alpha}^{t_{\alpha+1}} \drift(t)\,dt 
&\,\in\, \conv\Big( \tilde{F}_{\epsilon+\delta}\big(\ptraj(t_\alpha)\big)  \Big) \\
&\, =\, \tilde{F}_{\epsilon+\delta}\big(\ptraj(t_\alpha)\big)\\
&\, =\, \conv\big\{ \drift \, \big| \, \drift\in F(y), \, y\in \ball_{\epsilon+\delta} (x) \big\}.
\end{split}
\end{equation}
Then, from the ``Caratheodory's theorem'' \cite{Rock96}, there exist $n+1$ number, $\drift_1^{\alpha},\ldots, \drift_{n+1}^{\alpha}$, of vectors in $\tilde{F}_{\epsilon+\delta}\big(\ptraj(t_\alpha)\big)$, and non-negative constants $\theta_1^\alpha,\ldots  ,\theta_{n+1}^\alpha$ with $\theta_1^\alpha+ \cdots +\theta_{n+1}^\alpha=1$ such that 
\begin{equation}\label{eq:int xi equals convex comb of z}
\frac{1}{t_{\alpha+1}-t_\alpha} \,\int_{t_\alpha}^{t_{\alpha+1}} \drift(t)\,dt 
\, = \, \sum_{i=1}^{n+1} \theta_i^\alpha\, \drift_i^\alpha.
\end{equation}
For any $i\le n+1$, let $z_i^\alpha$ be a point in $\ball_{\epsilon+\delta}\big(\ptraj(t_\alpha)\big)$ such that $\drift_i^\alpha  \in F(z_i^\alpha)$.
Also let $T_0^\alpha=t_\alpha$, and for any $i\le n+1$, 
\begin{equation}\label{eq:def Ti}
t_i^\alpha \,= \, (t_{\alpha+1}-t_\alpha)\sum_{j=1}^{i} \theta_i^\alpha.
\end{equation}

We now define the functions $\pert',\yt',\drift':\R_+\to\R^n$ as follows.
For any $\alpha\le \alphamax$, let
\begin{equation}\label{eq:def of xi , u yt for talpha}
\begin{split}
\pert'(t_\alpha)&=\pert(t_\alpha),\\ 
\yt(t_\alpha)&=\ptraj(t_\alpha),\\
\drift'(t_\alpha) & = \drift(t_\alpha).
\end{split}
\end{equation}
For any $\alpha\le \alphamax$, any $i\le n$, and any $t\in[T_{i-1}^\alpha,T_i^\alpha)$ excluding $t=t_\alpha$, let
\begin{equation}\label{eq:def yt, u, xi for non talpha}   
\begin{split}
\yt(t)  &\,=\, z_i^\alpha,\\
\drift'(t) &\, =\, \drift_i^\alpha,\\
\pert'(t)&\,=\,\pert(t_\alpha) \,+\, \big(z_i^\alpha -\ptraj(t_\alpha)\big)\, -  \, \int_{t_\alpha}^t \drift'(\tau)\,d\tau.
\end{split}
\end{equation}
Then, for any $t\ge 0$,
\begin{equation}\label{eq:drift' is ni F y tilde}
\drift'(t)\in F\big(\yt(t)\big).
\end{equation}

In the reset of the proof, we will show that \eqref{eq:yt near ptraj} and \eqref{eq:pp close to pert} hold, and that $\yt(\cdot)$ is a perturbed trajectory of $F(\cdot)$  corresponding to perturbation $\pert'(\cdot)$.

Since $z_i^\alpha\in\ball_{\epsilon+\delta}\big(\ptraj(t_\alpha)\big)$, for all $\alpha<\alphamax$ and all $i\le n+1$, it follows that
for any $t\in [T_{i-1}^\alpha,T_i^\alpha)$, 
\begin{equation}
\Ltwo{\yt(t)-\ptraj(t)} \, = \, \Ltwo{z_i^\alpha-\ptraj(t)}  \, \le\,    \Ltwo{z_i^\alpha-\ptraj(t_\alpha)}   +   \Ltwo{\ptraj(t_\alpha)-\ptraj(t)} \,\le\, (\epsilon+\delta) + \delta,
\end{equation} 
and \eqref{eq:yt near ptraj} is satisfies by a proper choice of $\delta$.
Moreover, it follows from \eqref{eq:tf t suc} and \eqref{eq:constant not to blow up} that for any $\alpha<\alphamax$ and any $i\le n+1$,
\begin{equation}
\begin{split}
\ltwo{(t_{\alpha+1}-t_\alpha)\,\drift_i^\alpha} &\, \le\, \frac{\delta}{\lip\,\big(\ltwo{\ptraj(t_{\alpha})}+\epsilon+\delta\big)+1}\, \ltwo{\drift_i^\alpha} \\
&\, \le\, \frac{\delta}{\lip\,\big(\ltwo{\ptraj(t_{\alpha})}+\epsilon+\delta\big)+1}\, \lip\,\ltwo{z_i^\alpha}\\
&\, \le\, \frac{\delta}{\lip\,\big(\ltwo{\ptraj(t_{\alpha})}+\epsilon+\delta\big)+1}\, \lip\,\big(\ltwo{\ptraj(t_{\alpha})}+\epsilon+\delta \big)\\
&\,<\delta.
\end{split}
\end{equation}
Then, for any $\alpha<\alphamax$, and $i\le n+1$, and  any $t\in [t_\alpha,t_{\alpha+1})$,
\begin{equation} \label{eq:inequality for size of u'}
\begin{split}
\Ltwo{\pert'(t)}\,&\le\, \Ltwo{\pert(t_\alpha)} \,+\, \Ltwo{z_i^\alpha -\ptraj(t_\alpha)}\, + \, \Ltwo{ \int_{t_\alpha}^t\drift'(\tau)\,d\tau},\\
&\le \Ltwo{\pert(t_\alpha)} + (\epsilon+\delta) + \delta.
\end{split}
\end{equation}
Therefore, 
\begin{equation} \label{eq:bound on the size of U' by the size of U}
\sup_{\tau \le t} \Ltwo{ \pp(\tau)} \,\le \, \sup_{\tau \le t}\Ltwo{ \pert(\tau)} \,+\, \epsilon \,+\, 2\delta,
\end{equation}
and \eqref{eq:pp close to pert 2} follows by a proper choice of $\delta$.

It only remains to show that $\yt(\cdot)$ is a perturbed trajectory of $F(\cdot)$  corresponding to perturbation $\pert'(\cdot)$.
\begin{clm}\label{claim:yt is a pert traj of u'}
\begin{itemize}
\item[a)]
For any $\alpha< \alphamax$,
\begin{equation}\label{eq:yt is traj transinite ind 1}
\yt(t_\alpha) \,=\, \yt(0) \,+\, \int_{0}^{t_\alpha} \drift'(\tau) \,d\tau\, + \pert'(t_\alpha).
\end{equation}

\item[b)]
For any $\alpha< \alphamax$, and any $t\in [t_\alpha, t_{\alpha+1})$,
\begin{equation}\label{eq:yt is traj transinite ind 2}
\yt(t) \,=\, \yt(t_\alpha) \,+\, \int_{t_\alpha}^t \drift'(\tau) \,d\tau\, + \pert'(t) - \pert'(t_\alpha).
\end{equation}
\end{itemize}
\end{clm}

\begin{proof}[Proof of Claim~\ref{claim:yt is a pert traj of u'}]
We first show that $\drift'(\cdot)$ is Lebesgue integrable. 
Since $\drift(\cdot)$ is a right-continuous  function of time, it is measurable.
Moreover, \eqref{eq:constant not to blow up} implies that for any $t\ge 0$, 
\begin{equation}\label{eq:sup xi bounded}
\sup_{\tau\in [0,T]} \Ltwo{\drift'(\tau)} \, \le\, \sup_{\tau\in [0,T]}   \lip\,\Ltwo{\yt(\tau)}
\,\le\, \sup_{\tau\in [0,T]}   \lip\,\big(\Ltwo{\ptraj(\tau)} + \epsilon+2\delta\big)
\,\le\,  \lip\,\big(e^{ct}  \Ltwo{0}  + \epsilon+2\delta\big)
\, <\,\infty.
\end{equation}
Therefore, $\drift'$ has finite integral over every bounded interval.
On the other hand, for any $\alpha<\alphamax$,
\begin{equation} \label{eq:integral of xi is equal to xi'}
\int_{t_\alpha}^{t_{\alpha+1}} \drift'(\tau)\, d\tau 
\, =\, (t_{\alpha+1}-t_{\alpha})\, \sum_{i=1}^{n+1} \theta_i^\alpha z_i^\alpha
\, =\, \int_{t_\alpha}^{t_{\alpha+1}} \drift(\tau)\, d\tau ,
\end{equation}
where the equalities are due to the definition of $\drift'$ and \eqref{eq:int xi equals convex comb of z}, respectively.

We now prove \eqref{eq:yt is traj transinite ind 1} via a transfinite induction \cite{HrbaJ99} on $\alpha\in\ord$.
\begin{itemize}[label=$\bullet$]
\item Base case:
$\yt(0) =\yt(0)+\pert'(0)$.

\item Induction step for the successor case: 
Consider a successor ordinal number $\alpha<\alphamax$, and suppose that \eqref{eq:yt is traj transinite ind 1}
holds for $\alpha-1$.
Then,
\begin{equation}
\begin{split}
\yt(0) \,+\, \int_{0}^{t_\alpha} \drift'(\tau) \,d\tau\, + \pert'(t_\alpha) \,
& =\, \left[\yt(0) \,+\, \int_{0}^{t_{\alpha-1}} \drift'(\tau) \,d\tau\, + \pert'(t_{\alpha-1}) \, \right]  \\
&\quad \,+\,
\int_{t_{\alpha-1}}^{t_\alpha} \drift'(\tau) \,d\tau \, +\, \big[\pert'(t_\alpha) - \pert'(t_{\alpha-1})\big] \\
& =\, \yt(t_{\alpha-1}) \,+\, \int_{t_{\alpha-1}}^{t_\alpha} \drift'(\tau) \,d\tau \, +\, \big[\pert'(t_\alpha) - \pert'(t_{\alpha-1})\big] \\
& =\, \ptraj(t_{\alpha-1})  \,+\, \int_{t_{\alpha-1}}^{t_\alpha} \drift(\tau) \,d\tau \, +\, \big[\pert(t_\alpha) - \pert(t_{\alpha-1})\big] \\
& =\,  \ptraj(t_{\alpha})\\
& =\,  \yt(t_{\alpha}).
\end{split}
\end{equation}
where the second equality is due to the induction hypothesis, the third equality is from \eqref{eq:def of xi , u yt for talpha} and \eqref{eq:integral of xi is equal to xi'}, the fourth equality is because $\ptraj(\cdot)$ is a perturbed trajectory corresponding to $\pert(\cdot)$, and the last equality is again from \eqref{eq:def of xi , u yt for talpha}.

\item Induction step for the limit case: 
Consider a limit ordinal number $\alpha<\alphamax$, and suppose that \eqref{eq:yt is traj transinite ind 1}
holds for  all ordinals $\beta< \alpha$. 
Then, for any $\beta< \alpha$,
\begin{equation} \label{eq:limit case for being traj}
\begin{split}
\yt(0) \,+\, \int_{0}^{t_\alpha} \drift'(\tau) \,d\tau\, + \pert'(t_\alpha) \,
& =\, \left[\yt(0) \,+\, \int_{0}^{t_{\beta}} \drift'(\tau) \,d\tau\, + \pert'(t_{\beta}) \, \right]   \\
&\quad \,+\,
\int_{t_{\beta}}^{t_\alpha} \drift'(\tau) \,d\tau \, +\, \big[\pert'(t_\alpha) - \pert'(t_{\beta})\big] \\
& =\, \yt(t_{\beta}) \,+\, \int_{t_{\beta}}^{t_\alpha} \drift'(\tau) \,d\tau \, +\, \big[\pert'(t_\alpha) - \pert'(t_{\beta})\big] \\
& =\, \ptraj(t_{\beta})  \,+\, \int_{t_{\beta}}^{t_\alpha} \drift(\tau) \,d\tau \, +\, \big[\pert(t_\alpha) - \pert(t_{\beta})\big] \\
&\quad \, +\,  \int_{t_{\beta}}^{t_\alpha} \big(\drift'(\tau)-\drift(\tau)\big) \,d\tau\\
& =\,  \ptraj(t_{\alpha}) \, +\,  \int_{t_{\beta}}^{t_\alpha} \big(\drift'(\tau)-\drift(\tau)\big) \,d\tau\\
& =\,  \yt(t_{\alpha}) \, +\,  \int_{t_{\beta}}^{t_\alpha} \big(\drift'(\tau)-\drift(\tau)\big) \,d\tau.
\end{split}
\end{equation}
where the second equality is due to the induction hypothesis, the third equality is from  \eqref{eq:def of xi , u yt for talpha} and \eqref{eq:integral of xi is equal to xi'}, the fourth equality is because $\ptraj(\cdot)$ is a perturbed trajectory corresponding to $\pert(\cdot)$, and the last equality is again from \eqref{eq:def of xi , u yt for talpha}.
As a result, $\int_{t_{\beta}}^{t_\alpha} \big(\drift'(\tau)-\drift(\tau)\big) \,d\tau$ is independent of choice of $\beta$.

It follows from \eqref{eq:sup xi bounded} that $\Ltwo{\drift'(t)-\drift(t)}$ is bounded for $t\le t_\alpha$.
Moreover, since $\alpha$ is a limit ordinal, from the definition \eqref{eq:tf t lim}, the sequence  $t_\beta$, for $\beta<\alpha$,  converges to $t_\alpha$ from below.
Then, $\Ltwo{\int_{t_{\beta}}^{t_\alpha} \big(\drift'(\tau)-\drift(\tau)\big) \,d\tau}$ becomes arbitrarily small for proper vaues of $\beta<\alpha$. Therefore, $\int_{t_{\beta}}^{t_\alpha} \big(\drift'(\tau)-\drift(\tau)\big) \,d\tau=0$, and \eqref{eq:yt is traj transinite ind 1} follows from \eqref{eq:limit case for being traj}.
\end{itemize}
This completes the proof of Part (a).

\medskip

For Part (b), recall the constants $t_i^\alpha$, $i=1,\ldots,n+1$, defined in \eqref{eq:def Ti}.
Then, for any  $i\le n+1$ and any $t\in [t_{i-1}^\alpha,t_i^\alpha)$,
\begin{equation}
\begin{split}
\yt(t_\alpha) \,+\, \int_{t_\alpha}^t\drift'(\tau)\,d\tau \,+\, \pert'(t) - \pert'(t_\alpha) 
&\, =\,  \yt(t_\alpha) \,+\, \int_{t_\alpha}^t\drift'(\tau)\,d\tau \,+\, \big(z_i^\alpha -\ptraj(t_\alpha)\big) \,-\, \int_{t_\alpha}^t\drift'(\tau)\,d\tau\\
& \, =\,  z_i^\alpha \\
& \, =\,  \yt(t) ,
\end{split}
\end{equation}
where the first and the last equalities are due to the definition \eqref{eq:def yt, u, xi for non talpha}. 
This completes the proof of the claim.
\end{proof}

Then, it follows from Parts (a) and (b) of claim \ref{claim:yt is a pert traj of u'}  that for any $t\ge 0$,
\begin{equation}\label{eq:yt is traj transinite ind 3}
\yt(t) \,=\, \yt(0) \,+\, \int_{0}^{t} \drift'(\tau) \,d\tau\, + \pert'(t).
\end{equation}
Together with \eqref{eq:drift' is ni F y tilde}, this implies that $\yt(\cdot)$ is a perturbed trajectory of $F$  corresponding to perturbation $\pert'$.

We finally note that $U'$ is right continuous everywhere, except for the times $t_\alpha$.
To satisfy right continuity also at times $t_\alpha$, we modify the definitions of $U'$, $\tilde{y}$, and $\xi'$, by eliminating \eqref{eq:def of xi , u yt for talpha} and considering \eqref{eq:def yt, u, xi for non talpha} also at $t=t_\alpha$. 
It is straightforward to see that this modification does not impact any of the integrals, and \eqref{eq:bound on the size of U' by the size of U} and  \eqref{eq:yt is traj transinite ind 3} would still be valid.
This completes the proof of Part (a) of the Lemma.

\medskip

The proof of  Part (b) is similar to the proof of Part (a). The only difference is  the choice of $t_\alpha$ in for successor ordinals $\alpha$ in the transfinite recursion. 
Here, we replace \eqref{eq:tf t suc} with
\begin{equation}
\begin{split}
t_\alpha\, = \, &\min\bigg( t_{\alpha-1}+ \frac{\delta}{\lip\,\big(\ltwo{\ptraj_1(t_{\alpha-1})} +\ltwo{\ptraj_2(t_{\alpha-1})}+\epsilon+\delta\big)+1}\, ,\\ 
&\qquad\qquad \sup\left\{t\in\R\,\Big|\, \ptraj_i(\tau) \in \ball_\delta\big(\ptraj_i({t_{\alpha-1}})\big)  ,\,\forall \tau\in[t_{\alpha-1},t),\, i=1,2 \right\}    \bigg).
\end{split}
\end{equation}
The only point here is that we use the same sequence $t_\alpha$ for both trajectories.
Then, instead of \eqref{eq:inequality for size of u'} we can write
\begin{equation} 
\begin{split}
\Ltwo{\pert'_2(t) -\pert'_1(t) }\,&\le\, \Ltwo{\pert_1(t_\alpha) - \pert_2(t_\alpha)} \,+\, \Ltwo{z_{i,1}^\alpha -\ptraj(t_\alpha)}\, + \, \Ltwo{z_{i,2}^\alpha -\ptraj(t_\alpha)} \\
&\quad\, + \, \Ltwo{ \int_{t_\alpha}^t\drift'_1(\tau)\,d\tau}\,+\, \Ltwo{ \int_{t_\alpha}^t\drift'_2(\tau)\,d\tau},\\
&\le\, \Ltwo{\pert_1(t_\alpha) - \pert_2(t_\alpha)} + 2(\epsilon+\delta) + 2\delta\\
& = \, \Ltwo{\pert_1(t_\alpha) - \pert_2(t_\alpha)} + 2\epsilon + 4\delta.
\end{split}
\end{equation}
This implies \eqref{eq:pp close to pert 2} for a proper choice of $\delta$, and completes the proof of Lemma~\ref{lem:simul spread}.

\medskip
\section{\bf Proof of Lemma~\ref{lem:completeness of linear} }  \label{app:proof sol linear}
First note that, by definition, if $y(\cdot)$ is  a perturbed trajectory, then for any $t\ge0$
\begin{equation} \label{eq:pert sol def for  linear proof}
y(t) \,-\, \int_0^t Ay(\tau)\,d\tau \,=\, y(0) \,+\, U(t) .
\end{equation}
By multiplying both sides with $e^{-At}$ and integration, we get  
\begin{equation}\label{eq:integral formula for y linear}
\int_0^t e^{-A\tau} y(\tau)\,d\tau \,-\, \int_0^t A e^{-A\tau}\int_0^\tau y(s)\,ds\,d\tau \,=\, y(0) \, \int_0^t e^{-A\tau}\,d\tau \, +\,\int_0^t e^{-A\tau} U(\tau) d\tau.
\end{equation}
From integration by parts (e.g., Theorem 12.5 in \cite{Gord94}),
\begin{equation*}
\int_0^t e^{-A\tau} y(\tau)\,d\tau \,-\, \int_0^t A e^{-A\tau}\int_0^\tau y(s)\,ds\,d\tau \,=\, e^{-At}\int_0^t y(\tau)\,d\tau.
\end{equation*}
Plugging this into the left hand side of \eqref{eq:integral formula for y linear}, we obtain 
\begin{equation}\label{integral formula for x linear}
\begin{split}
\int_0^t y(\tau)\,d\tau &\,=\, y(0) \, e^{At} \int_0^t e^{-A\tau}\,d\tau \, +\,e^{At} \int_0^t e^{-A\tau} U(\tau) d\tau\\
&\,=\, y(0) \, \int_0^t e^{A\tau}\,d\tau \, +\,e^{At} \int_0^t e^{-A\tau} U(\tau) d\tau.
\end{split}
\end{equation}

We now show that $x(\cdot)$ defined in \eqref{eq:closed form formul for pert traj of linear sys} is a solution of \eqref{integral formula for x linear}. We have
\begin{equation}
\begin{split}
\int_0^t x(\tau)\,d\tau \,&=\, x(0) \int_0^t  e^{A\tau}\,d\tau  \,+\, \int_0^t U(\tau)\,d\tau \,+\, \int_0^t A e^{A\tau} \int_{0}^\tau e^{-As}U(s)\,ds\,d\tau\\
&=\, x(0) \int_0^t  e^{A\tau}\,d\tau  \,+\, \int_0^t U(\tau)\,d\tau \,+\, \int_0^t \left(\int_{s}^t A e^{A\tau} \,d\tau\right)\, e^{-As}U(s)\,ds\\
&=\, x(0) \int_0^t  e^{A\tau}\,d\tau  \,+\, \int_0^t U(\tau)\,d\tau \,+\, \int_0^t \big( e^{At} -e^{As}\big)\, e^{-As}U(s)\,ds\\
&=\, x(0)  \int_0^t e^{A\tau}\,d\tau \, +\,e^{At} \int_0^t e^{-A\tau} U(\tau) d\tau.
\end{split}
\end{equation}
Then, $x(\cdot)$ satisfies \eqref{integral formula for x linear}. 
Moreover, for any other solution $x'(\cdot)$ of \eqref{integral formula for x linear}, $\int_0^t \big(x'(\tau)-x(\tau)\big)=0$, for all $t\ge0$. Then, $x'(\cdot)$ equals $x(\cdot)$, almost everywhere.
Therefore, for any solution $y(\cdot)$ of \eqref{eq:pert sol def for  linear proof},
\begin{equation} 
\begin{split}
y(t) &\,=\, \int_0^t Ay(\tau)\,d\tau \,+\, y(0) \,+\, U(t) \\
&\,= \, \int_0^t Ax(\tau)\,d\tau \,+\, y(0) \,+\, U(t) \\
&\,= \, x(t).
\end{split}
\end{equation}
Therefore, $x(\cdot)$ is the unique perturbed trajectory corresponding to the perturbation function $U(\cdot)$, and the lemma follows.

\end{document}